\pdfoutput=1
\documentclass[aps,pra,10pt,twocolumn,notitlepage,noeprint,superscriptaddress,nofootinbib]{revtex4-1}
\usepackage[utf8]{inputenc}
\usepackage{microtype}
\usepackage{amsmath,amsfonts,amssymb,amstext,bbm}
\usepackage{mathtools}
\usepackage{amsthm}
\usepackage{xcolor}
\usepackage{hyperref}
\hypersetup{
	breaklinks=true,
	colorlinks=true,
	linkcolor=blue,
	urlcolor=blue,
	citecolor=blue
}

\newcommand{\id}{\mathbbm{1}}

\newcommand{\I}{\mathrm{i}}
\newcommand{\T}{\top}
\let\Re\relax
\let\Im\relax

\DeclareMathOperator{\Re}{Re}
\DeclareMathOperator{\Im}{Im}
\DeclareMathOperator{\Tr}{Tr}
\DeclareMathOperator{\tr}{tr}

\renewcommand{\vec}{\boldsymbol}
\let\originalleft\left
\let\originalright\right
\renewcommand{\left}{\mathopen{}\mathclose\bgroup\originalleft}
\renewcommand{\right}{\aftergroup\egroup\originalright}
\renewcommand{\right}{\aftergroup\egroup\originalright}

\newcommand{\parti}[2]{\frac{\partial #1}{\partial #2}}

\newcommand{\avg}[1]{\langle#1\rangle}
\newcommand{\Avg}[1]{\left\langle#1\right\rangle}
\newcommand{\abs}[1]{\left|#1\right|}
\newcommand{\bk}[1]{\left(#1\right)}
\newcommand{\Bk}[1]{\left[#1\right]}
\newcommand{\BK}[1]{\left\{#1\right\}}

\newcommand{\norm}[1]{\lVert #1 \rVert}
\newcommand{\deff}{\vec{X}_{\mathrm{eff}}}
\DeclareMathOperator{\spn}{span}
\DeclareMathOperator{\cspn}{\overline{span}}
\DeclareMathOperator{\trace}{Tr}
\DeclareMathOperator{\real}{Re}
\DeclareMathOperator{\imag}{Im}

\newtheorem{theorem}{Theorem}
\newtheorem{corollary}{Corollary}
\newtheorem{proposition}{Proposition}
\newtheorem{lemma}{Lemma}

\begin{document}
\title{Upper bounds on the Holevo Cramér-Rao bound for multiparameter quantum parametric and semiparametric estimation}

\author{Francesco Albarelli}
\email{francesco.albarelli@gmail.com}
\affiliation{Faculty of Physics, University of Warsaw, 02-093 Warszawa, Poland}
\affiliation{Department of Physics, University of Warwick, Coventry CV4 7AL, United Kingdom}

\author{Mankei Tsang}
\email{mankei@nus.edu.sg}
\homepage{https://blog.nus.edu.sg/mankei/}
\affiliation{Department of Electrical and Computer Engineering,
  National University of Singapore, 4 Engineering Drive 3, Singapore
  117583}
\affiliation{Department of Physics, National University of Singapore,
  2 Science Drive 3, Singapore 117551}

\author{Animesh Datta}
\email{animesh.datta@warwick.ac.uk}
\affiliation{Department of Physics, University of Warwick, Coventry CV4 7AL, United Kingdom}

\date{\today}




\begin{abstract}
We formulate multiparameter quantum estimation in the parametric and semiparametric setting.
While the Holevo Cramér-Rao bound (CRB) requires no substantial modifications in moving from the former to the latter, we generalize the Helstrom CRB appropriately.
We show that the Holevo CRB cannot be greater than twice the generalized Helstrom CRB.
We also present a tighter, intermediate, bound.
Finally, we show that for parameters encoded in the first moments of a Gaussian state there always exists a Gaussian measurement that gives a classical Fisher information matrix that is one-half of the quantum Fisher information matrix.
\end{abstract}

\maketitle

\section{Introduction}

In recent years, there has been a renewed interest in multiparameter quantum estimation stemming from efforts to deliver quantum-enhanced sensing devices~\cite{Humphreys2013,Vidrighin2014,Baumgratz2015,Szczykulska2016,Gagatsos2016a,Chrostowski2017,Pezze2017,Roccia2017,Roccia2018,Gessner2018,Yang2018b,Genoni2019,Gorecki2019,Rubio2019,Bisketzi2019,Liu2019d,Lu2019,Albarelli2019c}.
Multiparameter quantum estimation was initiated by Helstrom~\cite{Helstrom1967,helstrom1976quantum}, who derived a quantum version of the classical Cramér-Rao bound (CRB) on the mean square error matrix of an estimator.
The scalar Helstrom CRB is defined as the weighted trace of the inverse of the quantum Fisher information matrix (QFIM).
However, its evaluation requires solving a Lyapunov equation corresponding to a density matrix and inverting the QFIM, neither of which are known to be possible, in general, analytically.
Furthermore, the attainability of such a bound is not guaranteed due to the non-commutativity of observables in quantum mechanics.
Tighter attainable bounds have therefore been sought and identified since~\cite{Yuen1973,Belavkin1976,Holevo1976,Holevo2011b}.

The fundamental asymptotically attainable \emph{scalar} bound on the weighted trace of the mean square error matrix is the so-called Holevo CRB~\cite{Holevo1976,Holevo2011b,Nagaoka1989,Hayashi2008a,Yamagata2013,Yang2018a}.
Evaluating the Holevo CRB requires an optimisation over a set of Hermitian operators which is not known to be possible analytically except in a few non-trivial cases~\cite{Holevo1976,Suzuki2016a,Bradshaw2017a,Bradshaw2017,Suzuki2018,Sidhu2019a}, leaving it largely unscrutinised.
It was recently shown that the evaluation of the Holevo CRB for finite-dimensional systems is a convex optimisation problem solvable via semidefinite programming~\cite{Albarelli2019}.

In this work, we consider the general problem of estimating a $q$-dimensional function $\vec{\beta}(\vec{\theta})$ of $p$ unknown parameters $\vec{\theta}$, where $q$ is finite but we can have $p=\infty$.
Most of the existing literature on multiparameter quantum estimation is restricted to $p<\infty$ and $\vec{\beta}(\vec{\theta}) = \vec{\theta}$. 
Following the terminology of classical statistics, parametric estimation covers $p<\infty$ while semiparametric estimation $p=\infty,$ even though the techniques developed for the latter case can be applied to the former.
The $p<\infty$ scenario has been studied recently using an \emph{a priori} partitioning of the the parameters into those of interest and another set of so-called nuisance parameters~\cite{Suzuki2019a}.
Our approach is more direct, requiring no \emph{a priori} partitioning, and necessary for genuine semiparametric estimation.

Our first result generalizes the quantum semiparametric estimation of Ref.~\cite{Tsang2019} to the multiparameter setting.
First, we introduce the appropriate local unbiasedness conditions for this setting. 
Then, without assuming that the QFIM is nonsingular, we derive a generalized Helstrom CRB that works both for parametric and semiparametric estimation.
For $p<\infty$, it can be expressed by using the Moore-Penrose pseudoinverse of the QFIM.
The standard Helstrom CRB is recovered for a nonsingular QFIM.
We then consider the Holevo CRB in this semiparametric setting.
Its evaluation remains unchanged from the parametric case, but the appropriate local unbiasedness conditions have to be considered.

Our second result proves that the Holevo CRB is never greater that twice the generalised Helstrom CRB.
Indeed, we also point out a tighter intermediate upper bound that can be computed from the same optimal argument needed to evaluate the generalized Helstrom CRB.
For $p<\infty$ this intermediate bound can be computed from the symmetric logarithmic derivatives (SLDs). 
Both bounds are saturable.

Our final result shows that for parameters encoded in the first moments of a Gaussian state (often called a Gaussian shift model) there exists a Gaussian measurement whose classical Fisher information matrix (FIM) equals one-half the QFIM (these are known as Fisher symmetric POVMs~\cite{Li2016g}); this is derived for parametric estimation ($p<\infty$) only.

Our second result implies that the Holevo CRB cannot provide additional information about possible quantum enhancements in scaling in multiparameter estimation not already available from the Helstrom CRB.
Nevertheless, the  Helstrom CRB may be inadequate in some cases and the Holevo CRB necessary.
An instance is judging a given quantum state's performance in applications such as simultaneous phase and loss estimation in optical interferometry~\cite{Albarelli2019}.
Moreover, for pure states the Holevo CRB is attainable without the need for collective measurements on multiple copies and its evaluation also identifies an optimal measurement that attains the bound~\cite{Matsumoto2002}.
Thus, for pure states the evaluation of the Holevo CRB is a particularly meaningful task when results of Ref.~\cite{Pezze2017} for the QFIM are not applicable.
On the other hand, the evaluation of the Holevo CRB for mixed states, while in general not allowing to identify an optimal measurement, should nonetheless provide a deeper quantitative understanding of the role of collective quantum measurements~\cite{inprep}.

Our final result, in conjunction with quantum local asymptotic normality~\cite{Kahn2009,Yamagata2013,Yang2018a}, suggests that a mean square error matrix equaling one-half the QFIM may be attainable asymptotically for arbitrary quantum statistical models, a statement that we leave as a conjecture.


\emph{Note added.}
While completing this work we became aware of an independent alternative derivation of inequality~\eqref{sandwich}, but only for nonsingular finite-dimensional models, by Carollo \emph{et al.}~\cite{Carollo2019}.
We have also found a small gap in their derivation, which we close in Appendix~\ref{app:Carollo}.




\subsection*{Notation}
We assume a quantum system described by an Hilbert space $\mathsf{H}$ and we denote the Hilbert space of Hermitian operators acting on it as $\mathcal{L}_\mathsf{h}(\mathsf{H})$.
We express the positive semidefinitess of an operator (or matrix) as $A \succeq 0$.
A state of the system is described by a density operator $\rho \in \mathcal{L}_\mathsf{h}(\mathsf{H}): \, \rho \succeq 0, \,\, \Tr \rho = 1$.
For clarity, we keep the notation for the trace in Hilbert space $\Tr$ separate from the trace of matrices $\tr$.

We use a vectorial notation for collection of operators $\vec{X} = (X_1,\dots,X_n)^\top$, where $\top$ denotes the transpose; all vectors are assumed to be column vectors.
Similarly we collect the partial derivatives of an operator $A(\vec{\theta})$ in a vector $\partial_{\vec{\theta}} A = (\partial_1 A(\vec{\theta}),\dots,\partial_p A(\vec{\theta}))^\top$.
To simplify the notation, we define the following:
\begin{enumerate}
\item The Jordan product of two operators is
\begin{align}
X\circ Y \equiv \frac{1}{2}\bk{XY+YX}.
\end{align}
\item For two vectors of operators $\vec{X} = (X_1,\dots,X_q)^\top$
and $\vec{Y} = (Y_1,\dots,Y_q)^\top$,
\begin{align}
\vec{X}^\top \circ \vec{Y} &\equiv \sum_s X_s \circ Y_s.
\end{align}
\item For two vectors of operators
$\vec{X} = (X_1,\dots,X_p)^\top$ and $\vec{Y} = (Y_1,\dots,Y_q)^\top$,
$\vec{X}\circ \vec{Y}^\top$ is a $p\times q$ matrix of operators given by
\begin{align}
(\vec{X}\circ \vec{Y}^\top)_{st} &\equiv X_s \circ Y_t.
\end{align}
\item A real-valued inner product between two vectors of Hermitian operators is given by
\begin{align}
\Avg{\vec{X},\vec{Y}} &\equiv \Tr \rho \vec{X}^\top \circ \vec{Y} = 
\Tr \rho \sum_{s=1}^q X_s \circ Y_s;
\label{inner}
\end{align}
for $q=1$ this is known as the symmetric logarithmic derivative (SLD) inner product~\cite{Hayashi2017c}.
\item The operator norm is
\begin{align}
\norm{\vec{X}} &\equiv \sqrt{\Avg{\vec{X},\vec{X}}}.
\end{align}
\item 
The Hilbert space of all zero-mean vectors of Hermitian operators\footnote{The Hilbert space of square summable operators was introduced by Holevo for $q=1.$
If $\rho$ is not full rank each element is actually an equivalence class. 
This observation is not important for our treatment here, but for example it allows to reduce the number of variables to optimize in the evaluation of the Holevo CRB for rank-deficient density matrices~\cite{Albarelli2019}.} is
\begin{align}
\mathcal H &\equiv \BK{\vec{X}: X_s \in \mathcal{L}_{\mathsf{h}} ( \mathsf{H}) , \norm{\vec{X}} < \infty, \Tr \rho \vec{X} = 0}.
\end{align}
\item A complex covariance matrix of a vector of operators $\vec{X} \in \mathcal H$
\begin{align}
\label{eq:Zdef}
Z(\vec{X}) &\equiv \Tr \rho \vec{X} \vec{X}^\top,
\end{align}
i.e. the Gram matrix of $\vec{X}$ w.r.t. the so-called right logarithmic derivative (RLD) inner product~\cite{Hayashi2017c}.
\item
A real covariance matrix of a vector of operators $\vec{X} \in \mathcal H$
\begin{align}
\label{eq:Vdef}
V(\vec{X}) &\equiv \Tr \rho \vec{X} \circ \vec{X}^\top = \Re Z(\vec{X}),
\end{align}
i.e. the Gram matrix of $\vec{X}$ w.r.t. the so-called symmetric logarithmic derivative (SLD) inner product.
\end{enumerate}

\section{Quantum parametric and semiparametric estimation}

A quantum statistical model is a mapping from the parameter space $\Theta \subset \mathbb{R}^{p}$ to density operators $\rho_{\vec{\theta}}$, where 
\begin{align}
\vec{\theta} = \bk{\theta_1,\dots,\theta_p}^\top \in \Theta
\end{align}
represents the $p$-dimensional vector of real parameters.
We do not assume a parametric model ($p<\infty$), but we will derive more explicit expressions for this case.
Additionally, we do not assume the derivatives $\partial_i \rho_{\vec{\theta}} \equiv \partial \rho_{\vec{\theta}}/\partial \theta_i$ to be linearly independent, allowing the possibility of a singular quantum statistical model.
We consider the general case in which the parameters of interest are represented by a $q$-dimensional function of the original unknown parameters
\begin{align}
\vec{\beta}(\vec{\theta}) = \bk{\beta_1(\vec{\theta}),\dots,\beta_q(\vec{\theta})}^\top
\in \mathbb R^q,
\end{align}
with $q \leq p$.
Here, we assume that both functions $\rho_{\vec{\theta}}$ and $\vec{\beta}(\vec{\theta})$ are sufficiently smooth.
We also assume that the rank of $\rho_{\vec{\theta}}$ is fixed, to avoid pathologies~\cite{Seveso2019}.
Furthermore we assume that the $p{\times}q$ matrix $\partial_{\vec{\theta}} \vec{\beta}$ with elements\footnote{This is the transpose of the Jacobian matrix of $\vec{\beta}(\vec{\theta})$ as conventionally defined.}
\begin{align}
\parti{\beta_s(\vec{\theta})}{\theta_j} \equiv \bk{\partial_{\vec{\theta}} \vec{\beta}}_{js},
\label{eq:partial_beta}
\end{align}
has rank $q$.
Note that most of the existing literature on multiparameter quantum estimation is restricted to $p<\infty$, $\vec{\beta}(\vec{\theta}) = \vec{\theta}$, $\partial_{\vec{\theta}}\vec{\beta}=\id_p$, as well as linearly independent $\partial_i \rho_{\vec{\theta}}$.

\subsection{Local unbiasedness conditions for semiparametric estimation}
A measurement on a quantum system is mathematically described by a positive operator valued measure (POVM) $M$, with outcome space\footnote{For a precise measure theoretical definition see e.g.~\cite{Heinosaari2011a}.} $\Omega$.
The $\vec{\theta}$-dependent probability of random outcomes is given by the Born rule $p_{\vec{\theta}}(x) = \Tr \rho_{\vec{\theta}} M(x)$.

Let $\vec{\check\beta}$ be an unbiased estimator that satisfies
\begin{align}
\label{eq:unbias}
\int \vec{\check \beta} (x)\trace \rho_{\vec{\theta}} M(d x) &= \vec{\beta}(\vec{\theta}),
\end{align}
where the integration is always understood to be on $\Omega$.
Following the terminology of semiparametric estimation, we define a vector of influence operators as
\begin{align}
\vec{X}(\vec{\theta}) &\equiv \int 
\Bk{\vec{\check\beta}(x)-\vec{\beta}(\vec{\theta})} M(dx).
\end{align}
By expanding the unbiasedness condition~\eqref{eq:unbias} around the true parameter value, we get the following \emph{local} unbiasedness conditions
\begin{align}
\label{eq:locunbcomp}
\trace \rho_{\vec{\theta}} \vec{X}_s(\vec{\theta}) &= 0,
&
\trace \parti{\rho_{\vec{\theta}}}{\theta_j} X_s(\vec{\theta})
&= \parti{\beta_s(\vec{\theta})}{\theta_j},
\end{align}
which we abbreviate as
\begin{align}
\label{eq:locunbvec}
\trace \rho_{\vec{\theta}} \vec{X} &= 0,
&
\trace (\partial_{\vec{\theta}} \rho_{\vec{\theta}}) \vec{X}^\top &= \partial_{\vec{\theta}} \vec{\beta},
\end{align}
where $\partial_{\vec{\theta}} \vec{\beta}$ is the $p{\times}q$ matrix defined in~\eqref{eq:partial_beta}.
Here and in the following, we assume that all functions are evaluated at the true parameter value $\vec{\theta}$.

The POVM and the classical estimator $(M,\vec{\check\beta})$ are collectively called a locally unbiased measurement.
The accuracy of the estimate is quantified by the mean square error matrix, which coincides with the covariance matrix for (locally) unbiased estimators
\begin{align}
\Sigma(\vec{\theta}) &\equiv \int 
\Bk{\vec{\check\beta}(x)-\vec{\beta}(\vec{\theta})} 
\Bk{\vec{\check\beta}(x)-\vec{\beta}(\vec{\theta})}^\top \trace \rho_{\vec{\theta}} M(dx).
\end{align}
We remark that using the covariance matrix of locally unbiased estimators for $\vec{\beta}$ as a figure of merit is equivalent to considering the classical CRB~\cite{lehmann_theory_1998}, given by the inverse Fisher information matrix (FIM) for the POVM $M$, as shown in Ref.~\cite{Nagaoka1989}.
The FIM for a probability distribution $p_{\vec{\theta}}(x)$ is defined as 
\begin{equation}
\label{eq:classFIM}
F ( p_{\vec{\theta}} ) \equiv \int [ \partial_{\vec{\theta}}  \log p_{\vec{\theta}}(x)] [ \partial_{\vec{\theta}} \log p_{\vec{\theta}}(x) ]^\top p_{\vec{\theta}}(x) dx .
\end{equation}

\subsection{Generalized Helstrom CRB}
Quantum estimation theory is based on the SLD operators defined by
\begin{align}
\label{eq:SLDdef}
\parti{\rho_{\vec{\theta}}}{\theta_j} &= \rho_{\vec{\theta}} \circ L_j.
\end{align}
We collect them in a vector $\vec{L}.$
The QFIM $J$, also called Helstrom information matrix, is obtained from the SLDs as
\begin{align}
\label{eq:QFIM}
J &\equiv V(\vec{L}) = \trace \rho_{\vec{\theta}} \vec{L} \circ \vec{L}^\top.
\end{align}

We introduce the set $\mathcal{X}_{\vec{\theta}} \subset \mathcal H$ of influence operators that obey the local unbiasedness conditions:
\begin{align}
\label{eq:influenceX}
\mathcal{X}_{\vec{\theta}} &\equiv \BK{\vec{X}: \vec{X} \in \mathcal H
\textrm{ and }
\trace (\partial_{\vec{\theta}} \rho) \vec{X}^\top = \partial_{\vec{\theta}} \vec{\beta}}.
\end{align}
The influence operator of any given locally unbiased measurement must be in $\mathcal{X}_{\vec{\theta}}$.
The conditions $\trace (\partial_{\vec{\theta}} \rho) \vec{X}^\top = \partial_{\vec{\theta}} \vec{\beta}$ can be also expressed in terms of the SLDs as $\trace \vec{L} \circ \vec{X}^\top = \partial_{\vec{\theta}} \vec{\beta}$, applying their definition~\eqref{eq:SLDdef}.
In the following we assume that the set $\mathcal{X}_{\vec{\theta}}$ is not empty, otherwise some parameters of interest cannot be estimated; we will give conditions for this hold in the case $p < \infty$ in Theorem~\ref{lem_nonemptyInfluence}.

We now present a pair of lemmas that are needed for our first result in Theorem~\ref{thm_ghb}, which is to generalise the single parameter semiparametric Helstrom CRB of Ref.~\cite{Tsang2019} to multiple parameters.
\begin{lemma}[Ref.~\cite{Holevo2011b}]
  Let $\Sigma$ be the error covariance matrix of a locally unbiased
  measurement and $\vec{X} \in \mathcal{X}_{\vec{\theta}}$ be the influence operator
  with respect to the measurement. Then
\begin{align}
\Sigma &\succeq V(\vec{X}).
\end{align}
\label{lem_V}
\end{lemma}
\begin{proof}
Delegated to Appendix~\ref{app:lem_V}. 
\end{proof}
We define the tangent space as
\begin{align}
\mathcal T &= \bk{\cspn \BK{\vec{L}}}^q \subseteq \mathcal H.
\end{align}
$\mathcal T$ is called a replicating space~\cite{tsiatis06}. 
Each entry of an $\vec{h} \in \mathcal T$ obeys
\begin{align}
h_s &\in \cspn\BK{\vec{L}}\quad
\forall \vec{h} \in \mathcal T;
\label{hs}
\end{align}
clearly the linear span is closed ($\cspn \BK{\vec{L}} = \spn \BK{\vec{L}}$) when $p < \infty$.
We also define the orthocomplement of $\mathcal T$ with respect to $\mathcal H$ as
$\mathcal T^\perp$.

\begin{lemma}[Strong orthogonality condition for a replicating space]
For any $\vec{h} \in \mathcal T$ and any $\vec{g} \in \mathcal T^\perp$,
\begin{align}
\trace \rho_{\vec{\theta}} \vec{h} \circ \vec{g}^\top &= 0.
\label{strong}
\end{align}
Conversely, any $\vec{g} \in \mathcal H$ that satisfies Eq.~\eqref{strong}
for all $\vec{h} \in \mathcal T$ must be in $\mathcal T^\perp$.
\label{lem_strong}
\end{lemma}
\begin{proof}
By definition,
\begin{align}
\Avg{\vec{h},\vec{g}} &= \trace \rho_{\vec{\theta}} \vec{h}^\top \circ \vec{g} = 0.
\end{align}
Take $h_s = h^{(1)}\delta_{st}$, where $h^{(1)}$ is an arbitrary
element of $\cspn\BK{\vec{L}}$.  Then
\begin{align}
\Avg{\vec{h},\vec{g}} &= \trace \rho_{\vec{\theta}} h^{(1)} \circ g_t = 0.
\end{align}
In other words, 
\begin{align}
g_t &\perp \cspn\BK{S}\quad \forall t,
\end{align}
which leads to Eq.~(\ref{strong}) via Eq.~(\ref{hs}). Conversely,
Eq.~(\ref{strong}) must imply
$\trace \rho_{\vec{\theta}} \vec{h}^\top \circ \vec{g} = \avg{\vec{h},\vec{g}} = 0$, so
$\vec{g} \in \mathcal T^\perp$ if Eq.~(\ref{strong}) holds for any
$\vec{h} \in \mathcal T$.
\end{proof}

Our first result is as follows.
\begin{theorem}[Generalized Helstrom CRB]
  For any locally unbiased measurement and any $q\times q$ real
  positive-semidefinite weight matrix $W \succ 0$,
\begin{align}
\tr W \Sigma &\ge \min_{\vec{X} \in\mathcal{X}_{\vec{\theta}}} \tr W V(\vec{X}) 
\equiv  C^{\mathrm{GS}} = \tr W V(\deff),
\label{ghb}
\end{align}
where
\begin{align}
\label{eq:X_eff}
\deff \equiv \Pi( \vec{X} |\mathcal T)
\end{align}
is the projection of any $\vec{X} \in \mathcal{X}_{\vec{\theta}}$ into the tangent
space $\mathcal T$.  Furthermore, $\deff$ is the unique minimizing
element in $\mathcal{X}_{\vec{\theta}}$.
\label{thm_ghb}
\end{theorem}
\begin{proof}
  The inequality in Eq.~(\ref{ghb}) follows from Lemma~\ref{lem_V}.
  To prove the second part of the theorem, let the orthogonal
  decomposition of an arbitrary $\vec{X} \in \mathcal{X}_{\vec{\theta}}$ be
\begin{align}
\vec{X} &= \deff + \vec{g},
&
\deff &= \Pi(\vec{X}|\mathcal T),
&
\vec{g} &= \vec{X}-\deff \in \mathcal T^\perp.
\end{align}
Using~\eqref{eq:Vdef} and Lemma~\ref{lem_strong} then leads to the strong Pythagorean
theorem \cite{tsiatis06}
\begin{align}
V(\vec{X}) &= \trace \rho_{\vec{\theta}} \deff \circ \deff^\top
+\trace \rho_{\vec{\theta}} \vec{g} \circ \vec{g}^\top + 
\nonumber\\&\quad
\trace \rho_{\vec{\theta}} \deff\circ \vec{g}^\top
+ \trace \rho_{\vec{\theta}} \vec{g} \circ \deff^\top
\\
&= V(\deff) + V(\vec{g}).
\end{align}
Since $V(\vec{g}) \succeq 0$,
\begin{align}
V(\vec{X}) &\succeq V(\deff),
&
\tr W V(\vec{X}) &\ge \tr W V(\deff).
\end{align}
To prove that $\deff$ is the unique minimizing element
in $\mathcal{X}_{\vec{\theta}}$, first note that
$\deff \in \mathcal{X}_{\vec{\theta}}$, because $\Pi(\vec{X}|\mathcal T)
\in \mathcal H$ and
\begin{align}
\trace (\partial_{\vec{\theta}} \rho_{\vec{\theta}}) \deff^\top
&= \trace \rho_{\vec{\theta}} \vec{L} \circ \deff ^\top
= \trace \rho_{\vec{\theta}} \vec{L} \circ \bk{\vec{X} -\vec{g}} ^\top
\\
&= \trace (\partial_{\vec{\theta}} \rho_{\vec{\theta}}) \vec{X}^\top- \trace \rho_{\vec{\theta}} \vec{L} \circ \vec{g}^\top  = 
\partial_{\vec{\theta}} \vec{\beta}.
\end{align}
Now suppose that $V(\vec{X}') = V(\deff)$ for another
$\vec{X}' \in \mathcal{X}_{\vec{\theta}}$. Write
$\vec{X}' = \deff + \vec{g}'$, where
$\vec{g}' = \vec{X}'-\deff$ can be shown to be in
$\mathcal T^\perp$ via Lemma~\ref{lem_strong} since
$\trace \rho_{\vec{\theta}} \vec{L} \circ \vec{g}'^\top = 0$.  The strong Pythagorean
theorem can then be applied to $\vec{X}'$, giving
$V(\vec{X}') = V(\deff) + V(\vec{g}')$.  If
$V(\vec{X}') = V(\deff)$, $V(\vec{g}') = 0$, and
$\norm{\vec{g}'}^2 = \trace V(\vec{g}') = 0$. Thus $\vec{g}'$ must be the zero element,
$\vec{X}' = \vec{X}_{\rm eff}$, and by contradiction
$\vec{X}_{\rm eff} = \Pi(\vec{X}|\mathcal T)$ must be the unique
minimizing element in $\mathcal{X}_{\vec{\theta}}$. In other words,
$\Pi(\vec{X}|\mathcal T)$ for any $\vec{X} \in \mathcal{X}_{\vec{\theta}}$ gives the
same $\vec{X}_{\rm eff}$.
\end{proof}

For parametric estimation $(p < \infty)$ we find the conditions on $J$ and $\partial_{\vec{\theta}} \vec{\beta}$ under which $\mathcal{X}_{\vec{\theta}}$ is not empty; this generalizes directly the classical results of~\cite{Stoica2001} to quantum estimation.
\begin{theorem}
  Let $p<\infty$, the set of influence operators $\mathcal{X}_{\vec{\theta}}$ is not empty if and only if all the columns of $\partial_{\vec{\theta}} \vec{\beta}$ are in the range of $J$, i.e. $J J^{+} \partial_{\vec{\theta}} \vec{\beta} = \partial_{\vec{\theta}} \vec{\beta}$, where $J^{+}$ denotes the Moore-Penrose pseudoinverse of $J$.
\label{lem_nonemptyInfluence}
\end{theorem}
\begin{proof}
It is an adaptation of Lemma 6.5.1 in Ref.~\cite{Holevo2011b} and is delegated to Appendix~\ref{app:lem_nonemptyInfluence}.
\end{proof}
Theorem~\ref{lem_nonemptyInfluence} means that to estimate all the parameters $\vec{\beta}(\vec{\theta})$ with a non-diverging variance when $J$ is singular the columns of $\partial_{\vec{\theta}} \vec{\beta}$ must be linear combinations of eigenvectors of $J$ corresponding to non-zero eigenvalues.
This is consistent with~\cite{Suzuki2019a}, where the singularity of $J$ is allowed only in the block pertaining to the nuisance parameters.	

In this case we provide a more explicit form for the generalized Helstrom CRB using the Moore-Penrose pseudoinverse of $J$.
\begin{corollary}[Parametric estimation]
When $p < \infty$ and the conditions $J J^{+} \partial_{\vec{\theta}} \vec{\beta} =  \partial_{\vec{\theta}} \vec{\beta}$ of Theorem~\ref{lem_nonemptyInfluence} are satisfied, where $J^{+}$ denotes the Moore-Penrose pseudoinverse, we have
\begin{align}
\deff &= \bk{\partial_{\vec{\theta}} \vec{\beta}}^\top J^{+} \vec{L},
\label{proj}
\\
V(\deff) &= \bk{\partial_{\vec{\theta}} \vec{\beta}}^\top J^{+}\partial_{\vec{\theta}} \vec{\beta},
\label{Veff}
\\
C^{\mathrm{GS}} &= \tr W \bk{\partial_{\vec{\theta}} \vec{\beta}}^\top J^{+} \partial_{\vec{\theta}} \vec{\beta} .
\label{CS}
\end{align}
\label{cor_hb}
\end{corollary}
\begin{proof}
Delegated to Appendix~\ref{app:cor_hb}.
\end{proof}

Another straight-forward corollary is obtained for a non-singular QFIM, recovering the well-known Helstrom CRB.
\begin{corollary}[Helstrom CRB~\cite{Helstrom1967,Hayashi2005}]
If $p < \infty$ and $J$ is not singular
\begin{equation}
C^\mathrm{GS} = \tr W \bk{\partial_{\vec{\theta}} \vec{\beta}}^\top J^{-1} \partial_{\vec{\theta}} \vec{\beta} \equiv C^{\mathrm{S}}
\end{equation}
\end{corollary}
\begin{proof}
For a nonsingular $J$ we have $J^{+}=J^{-1}$ in~\eqref{CS}.
\end{proof}
The equivalent definition of $C^{\text{S}}$ as a constrained quadratic minimization (cf.~\eqref{ghb}), without explicitly evaluating and inverting $J$, was presented and exploited in Ref.~\cite{Nagaoka1989,Ragy2016} for $\vec{\beta}(\vec{\theta})=\vec{\theta}$.
However, $C^{\text{GS}}$ is more general than $C^{\text{S}}$, because it can be applied to linearly dependent $\{\vec{L}\}$, for which $J^{-1}$ is undefined.
This is especially useful for semiparametric problems~\cite{Tsang2019}.

\subsection{Holevo CRB}
Unlike the Helstrom CRB, generalizing the Holevo CRB for multiparameter estimation to the semiparametric setting requires no substantial modification, since it does not rely on any matrix CRB.
The only relevant difference is to consider the appropriate set of influence operators~\eqref{eq:influenceX}.
Again, we start with a couple of preparatory lemmas.
\begin{lemma}[Ref.~\cite{Holevo2011b}]
  Let $\Sigma$ be the error covariance matrix of a locally unbiased
  measurement and $\vec{X} \in \mathcal{X}_{\vec{\theta}}$ be the influence operator
  with respect to the measurement. Then
\begin{align}
\Sigma &\succeq Z(\vec{X}).
\end{align}
\label{lem_Z}
\end{lemma}
\begin{proof}
Delegated to Appendix~\ref{app:lem_Z}.
\end{proof}
\begin{lemma}[Belavkin and Grishanin~\cite{Belavkin1973}]
Let $A \succeq 0$ be a positive-semidefinite matrix.
Then
\begin{align}
\trace \real A &\succeq \norm{\imag A}_1.
\end{align}
\label{lem_bg}
\end{lemma}
\begin{proof}
Delegated to Appendix~\ref{app:lem_bg}.
\end{proof}

We can now introduce the Holevo CRB.
\begin{theorem}[Holevo CRB~\cite{Holevo2011b,Nagaoka1989}]
\begin{align}
\tr W\Sigma &\ge \min_{\vec{X} \in \mathcal{X}_{\vec{\theta}}}
\Bk{\tr W \real Z(\vec{X}) + \norm{\sqrt{W}\imag Z(\vec{X})\sqrt{W}}_1} \label{hcrb} \\
& \equiv C^\mathrm{H}. \notag
\end{align}
\end{theorem}
\begin{proof}
From Lemma~\ref{lem_Z},
\begin{align}
\sqrt{W}\Bk{\Sigma -Z(\vec{X})}\sqrt{W} \succeq 0.
\end{align}
From Lemma~\ref{lem_bg},
\begin{align}
\tr \real\sqrt{W}\Bk{\Sigma - Z(\vec{X})}\sqrt{W} &\ge
\norm{\imag\sqrt{W}\Bk{\Sigma-Z(\vec{X})}\sqrt{W}}_1.
\end{align}
Since $W$ and $\Sigma$ are real, we obtain
\begin{align}
\tr W \Sigma &\ge \tr W \real Z(\vec{X}) + 
\norm{\sqrt{W}\imag Z(\vec{X})\sqrt{W}}_1
\ge C^\mathrm{H}.
\end{align}
\end{proof}

The Holevo CRB is always more informative than the generalized Helstrom CRB, which we now show explicitly. 
Furthermore, it is also more informative than the RLD CRB~\cite{Yuen1973,Belavkin1976}, which we are not considering in this work.
\begin{corollary}
The Holevo CRB is never smaller than the generalized Helstrom CRB
\begin{equation}
C^\mathrm{GS} \leq C^\mathrm{H}
\end{equation}
\end{corollary}
\begin{proof}
Let $\vec{X}_\mathrm{opt} \in \mathcal{X}_{\vec{\theta}}$ be an argument that achieves the minimization in Eq.~\eqref{hcrb}.
Then
\begin{align}
C^\mathrm{H} &= 
\tr W \real Z(\vec{X}_\mathrm{opt}) + \norm{\sqrt{W}\imag Z(\vec{X}_\mathrm{opt})\sqrt{W}}_1
\\
&\ge \tr W \real Z(\vec{X}_\mathrm{opt}) = \tr W V(\vec{X}_\mathrm{opt})
\\
&\ge \min_{\vec{X} \in \mathcal{X}_{\vec{\theta}}} W V(\vec{X}) = C^\mathrm{GS}.
\end{align}
\end{proof}

Finally, we mention an alternative definition of the Holevo CRB~\cite{Nagaoka1989,Hayashi2017c}
\begin{align}
C^\mathrm{H}=\min_{V \in \mathbb{S}^q, \vec{X} \in \mathcal \mathcal{X}_{\vec{\theta}}}
\Bk{\tr W V  \, \, \vert  \, \,  V \succeq Z(\vec{X})},
\end{align}
where $\mathbb{S}^q$ is the set of $q{\times}q$ real-valued symmetric matrices.
For $\dim \mathsf{H} < \infty$ and $p<\infty$, this form allows to easily recast the bound as a semidefinite program after choosing a basis for $\mathcal{L}_{\mathsf{h}}(\mathsf{H})$.
This is a straight-forward generalization of the method proposed in Ref.~\cite{Albarelli2019}, where only the case $\vec{\beta}=\vec{\theta}$ and $J \succ 0$ was considered.

\section{Upper bounds on the Holevo CRB}\label{sec:uppbound}
Our main results are the two inequalities in the following theorem.
\begin{theorem}
The Holevo CRB can be upper bounded as
\begin{align}
C^\mathrm{H} \le C^\mathrm{D} \le 2C^\mathrm{GS},
\label{sandwich}
\end{align}
where $C^\mathrm{D} \equiv \tr W \real Z(\deff) + 
\norm{\sqrt{W}\imag Z(\deff)\sqrt{W}}_1$ and $\deff$ is defined in~\eqref{eq:X_eff}.
\end{theorem}
\begin{proof}
The first inequality holds by definition~\eqref{hcrb}
\begin{align}
C^\mathrm{H} &\le 
\trace W \real Z(\deff) + 
\norm{\sqrt{W}\imag Z(\deff)\sqrt{W}}_1 = C^\mathrm{D} .
\end{align}
To obtain the second inequality note that
\begin{align}
\sqrt{W} Z \sqrt{W} \succeq 0,
\end{align}
so Lemma~\ref{lem_bg} gives
\begin{align}
\norm{\sqrt{W}\imag Z\sqrt{W}}_1
&= \norm{\imag \sqrt{W}Z\sqrt{W}}_1 
\\
&\le \tr \real \sqrt{W} Z \sqrt{W}= \tr W \real Z.
\end{align}
Hence,
\begin{align}
C^\mathrm{D} \le 2\tr W \real Z(\deff) = 2 \tr W V(\deff)
= 2C^\mathrm{GS}.
\end{align}
\end{proof}
We call $C^\mathrm{D}$ the D-invariant CRB, because it is equal to $C^\mathrm{H}$ for any $W \succ 0$ when the model is D-invariant~\cite{Holevo2011b,Suzuki2016a,Suzuki2018}; in this case it also equals the RLD scalar CRB.
Without delving into the details, D-invariance is a mathematical property that guarantees that the optimization in~\eqref{hcrb} can be restricted to $\mathcal{X}_{\vec{\theta}} \cap \mathcal{T}$, thus giving $\deff$ as the optimal argument.

Two observations on the bound $C^\mathrm{D}$ are in order.
First, it is tighter than $2 C^\text{GS}$, yet obtained from $\deff$ by solving the same optimization problem (or by finding the SLDs for $p<\infty$, see Proposition~\ref{cohCRB}).
Second, this upper bound can be a loose restriction on the difference $C^\text{H}-C^\text{GS}$, which can be small and cannot be estimated without explicitly evaluating $C^\text{H}$.
The problem of noisy 3D magnetometry with multi-qubit systems offers such an illustration~\cite{Albarelli2019}.

In the parametric case we can write $C^\mathrm{D}$ more explicitly, using the expression~\eqref{proj} for $\deff$.
This is a generalization of the inequality derived in Ref.~\cite{Suzuki2016a} for $J \succ 0$.
\begin{proposition}
\label{cohCRB}
For $p<\infty$
\begin{align}
C^\mathrm{D} =& \tr W  (\partial_{\vec{\theta}} \vec{\beta})^\top J^{+} (\partial_{\vec{\theta}} \vec{\beta}) \notag \\
& + \norm{\sqrt{W} (\partial_{\vec{\theta}} \vec{\beta})^\top J^{+} D J^{+} (\partial_{\vec{\theta}} \vec{\beta}) \sqrt{W}}_1 \\
=& C^\mathrm{GS} + \norm{\sqrt{W} (\partial_{\vec{\theta}} \vec{\beta})^\top J^{+} D J^{+} (\partial_{\vec{\theta}} \vec{\beta}) \sqrt{W}}_1 ,
\label{eq:CD}
\end{align}
where we have introduced the matrix $D \equiv \imag Z (\vec{L})$.
\end{proposition}
\begin{proof}
Using~\eqref{proj} for $\deff$ when $p<\infty$ we get $Z \left( \deff \right) = (\partial_{\vec{\theta}} \vec{\beta})^\top J^{+} Z(\vec{L}) J^{+} (\partial_{\vec{\theta}} \vec{\beta})$.
\end{proof}

Finally, we consider the case of a scalar function $\beta(\vec{\theta})$.
\begin{proposition}
If $\beta(\vec{\theta})$ is a scalar ($q = 1$) $C^\mathrm{H} = C^\mathrm{GS}$.
\label{prop_scalar}
\end{proposition}
\begin{proof}
If $q = 1$, $Z(\vec{X}) = \trace \rho_{\vec{\theta}} \vec{X}^2$ is a scalar and real, 
meaning that $\imag Z(\vec{X}) = 0$.
Then
\begin{align}
C^\mathrm{H} &= \min_{\vec{X} \in \mathcal{X}_{\vec{\theta}}} \tr W \real Z(\vec{X}) = 
\min_{\vec{X} \in \mathcal{X}_{\vec{\theta}}} \tr W V(\vec{X}) = C^\mathrm{GS}.
\end{align}
\end{proof}

Proposition~\ref{prop_scalar} is especially relevant to
Ref.~\cite{Tsang2019}, which employs the generalized Helstrom bound
$C^\mathrm{GS}$ for a scalar $\beta$ and a multidimensional $\vec{\theta}$.
As $C^\mathrm{H}$ is asymptotically attainable~\cite{Kahn2009,Yamagata2013,Yang2018a}
(at least when $p < \infty$), the proposition implies that the $C^\mathrm{GS}$
bounds computed in Ref.~\cite{Tsang2019} for semiparametric problems are also asymptotically attainable in the finite-dimensional case.




Note that the upper bounds~\eqref{sandwich} are tight and can be attained.
For example, quantum tomography of pure states is an estimation problem for which our bounds are saturated, that is, $C^{\mathrm{D}}=C^{\text{H}}=2 C^{\text{S}}$.
Li \emph{et al}.~\cite{Li2016g} have shown that, for the estimation of all the $2d-2$ parameters of a finite-dimensional pure state, there exists a POVM with a classical Fisher information matrix $F= \frac{1}{2} J$, irregardless of $d$.
This POVM attains the Gill-Massar inequality~\cite{Gill2000}, which is equivalent to attaining the Holevo CRB for pure states~\cite{Matsumoto2002} and implies our claim that $C^{\text{H}}=2 C^{\text{S}}$ for this model.

\section{Attaining one-half of the QFIM in Gaussian shift models}
In this section we restrict ourselves to $p<\infty$.
A Gaussian state $\rho_{\vec{\theta}}^{\text{G}}$ is the thermal state (ground state if pure) of a Hamiltonian quadratic in the canonical operators $\vec{r} = (x_1, p_1, \dots x_k, p_k)^\top$ for $k$ bosonic modes, satisfying $ [ r_i , r_j ] = \I \Omega_{ij}$ with $\Omega = \I \oplus_{i=1}^k \sigma_y $, with the Pauli matrix $(\sigma_y)_{11}=(\sigma_y)_{22}=0,\, (\sigma_y)_{21}=-(\sigma_y)_{12}=\I$.
In particular we consider the quantum statistical model where the parameter dependence is only in the first moments $\overline{\vec{r}}_{\vec{\theta}} \equiv \Tr \left[ \rho_G \vec{r} \right]$ of the Gaussian state, while the parameter-independent covariance matrix (CM) is $\sigma \equiv 2 V(\vec{r} -\overline{\vec{r}}_{\vec{\theta}}) =\Tr \left[ \rho_{\vec{\theta}}^{\text{G}} ( \vec{r} -\overline{\vec{r}}_{\vec{\theta}}) \circ ( \vec{r} -\overline{\vec{r}}_{\vec{\theta}})^\top \right] $.
A Gaussian measurement is defined by an arbitrary physical CM $\sigma_{\text{m}} \succeq \I \Omega $ and a vector of $2k$ outcomes $\vec{r}_{\text{out}}$, physically implemented by noisy general-dyne detection.
The corresponding probability distribution reads~\cite{serafini2017quantum,Genoni2016}
\begin{equation}
p_{\vec{\theta}}(\vec{r}_\text{out}) = \frac{ \exp\left[ - \left(\vec{r}_\text{out} - \overline{\vec{r}}_{\vec{\theta}} \right)^\top \left( \sigma + \sigma_{\text{m}} \right)^{-1} \left( \vec{r}_\text{out} - \overline{\vec{r}}_{\vec{\theta}} \right)\right]}{\pi^{k} \sqrt{\det \left( \sigma + \sigma_{\text{m}}  \right)}}
\end{equation}
and the associated classical FIM~\eqref{eq:classFIM} is
\begin{equation}
\label{eq:gaussianCFIM}
F ( \vec{r}_{\vec{\theta}}, \sigma, \sigma_{\text{m}} ) \equiv 2 \left( \partial_{\vec{\theta}} \vec{r}_{\vec{\theta}} \right)^\T \left( \sigma + \sigma_{\text{m}} \right)^{-1} \left( \partial_{\vec{\theta}} \vec{r}_{\vec{\theta}} \right),
\end{equation}
where we have introduced the $2k{\times}p$ Jacobian matrix with elements $\left(\partial_{\vec{\theta}} \vec{r}_{\vec{\theta}} \right)_{ij} =\partial (r_{\vec{\theta}})_i / \partial \theta_j $ and highlighted the dependence on the state through $\vec{r}_{\vec{\theta}}$ and $\sigma$ and on the measurements through $\sigma_{\mathrm{m}}$.
The QFIM for a Gaussian shift model has exactly the same form as the FIM for a classical Gaussian distribution~\cite{Monras2013,serafini2017quantum}
\begin{equation}
\label{eq:gaussianQFIM}
J(\vec{r}_{\vec{\theta}}, \sigma ) \equiv 2 \left( \partial_{\vec{\theta}} \vec{r}_{\vec{\theta}} \right)^\T \sigma^{-1} \left( \partial_{\vec{\theta}} \vec{r}_{\vec{\theta}} \right),
\end{equation}
where we highlighted the dependence on $\vec{r}_{\vec{\theta}}$ and $\sigma$ only.

Considering the measurement with the same CM as the state, that is ${\sigma + \sigma_{\text{m}} = 2 \sigma}$, we get our third result
\begin{equation}
\label{eq:GaussianCMfisher}
F  ( \vec{r}_{\vec{\theta}}, \sigma, \sigma ) = \frac{1}{2} J(\vec{r}_{\vec{\theta}}, \sigma ).
\end{equation}
Thus, the scalar CRB obtained from the (pseudo-)inverse classical Fisher information matrix is indeed twice the generalized Helstrom CRB $\tr W (\partial_{\vec{\theta}} \vec{\beta})^\top F ( \vec{r}_{\vec{\theta}}, \sigma, \sigma )^{+} (\partial_{\vec{\theta}} \vec{\beta}) = 2 C^{\mathrm{GS}}$.
This halving of the available information manifests the additional noise of measuring complementary observables simultaneously as suggested by Arthurs and Kelly~\cite{Arthurs1965,Stenholm1992,Raymer1994}.

\section{Discussion}

The theory of quantum local asymptotic normality (QLAN)~\cite{Hayashi2006a,Hayashi2008a,Guta2007,Guta2007,Kahn2009,Yamagata2013,Yang2018a,Fujiwara2018} maps an asymptotically large number of identical copies of a finite-dimensional quantum statistical model to a Gaussian shift model\footnote{This has been used to show the asymptotic attainability of the Holevo CRB, although the precise technical assumptions differ in the various formulations.}.
Since for such models there always exists a POVM with a FIM equal to one-half of the QFIM as shown by~\eqref{eq:GaussianCMfisher}, we have $C^\mathrm{H} \leq 2 C^\mathrm{GS}$, in consistence with~\eqref{sandwich}.

This result strongly suggests that for any quantum statistical model there will be a sequence $(M_n,\vec{\check \beta}_n)$ of measurements and classical estimators on $n$ copies (not necessarily unbiased for finite $n$) such that
\begin{equation} \lim_{n \to \infty} n \Sigma\left( \rho_{\vec{\theta}}^{\otimes n}, M_n,\vec{\check \beta}_n \right) = 2 (\partial_{\vec{\theta}} \vec{\beta})^\top J^{+}(\partial_{\vec{\theta}} \vec{\beta})
\end{equation}
However, due to the technical assumptions underlying QLAN, our suggestion does not constitute a rigorous proof, for which one would need to present one such sequence.
We stress that, on the contrary, the derivation of the scalar inequality $C^\text{H} \leq 2 C^\text{GS}$ is based purely on the \emph{evaluation} of the Holevo CRB, which is a well-defined ``single letter'' calculation, regardless of asymptotics.


\section*{Acknowledgments}
We thank M.~G.~Genoni both for several fruitful discussions and for making us aware of Ref.~\cite{Carollo2019}.
We are grateful to D.~Branford, R.~Demkowicz-Dobrzański, J.~F.~Friel, W.~Górecki and J.~Suzuki for useful discussions.
AD and FA have been supported by the UK EPSRC (EP/K04057X/2) and the UK National Quantum Technologies Programme (EP/M01326X/1, EP/M013243/1).
FA also acknowledges financial support from the National Science Center (Poland) grant No. 2016/22/E/ST2/00559.
MT is supported by the Singapore National Research Foundation under Project No. QEP-P7.

\appendix
\section{Proofs omitted in the main text}

\subsection{Proof of Lemma~\ref{lem_V}}
\label{app:lem_V}
\begin{proof}
Let $f(x)$ be a real function and $F = \int f(x) M(dx)$
be a Hermitian operator. Consider the inequality
\begin{align}
\int \Bk{f(x)-F} M(dx)\Bk{f(x)-F} &\succeq 0,
\end{align}
which leads to
\begin{align}
\int \Bk{f(x)}^2 M(dx)&\succeq F^2.
\end{align}
Let $u$ be an arbitrary real vector.
Substituting
\begin{align}
f(x) &= \mathbf{u}^\top \Bk{\vec{\check\beta}(x)-\vec{\beta}}
= \sum_s u_s\Bk{\check\beta_s(x)-\beta_s}
\end{align}
gives
\begin{align}
\mathbf{u}^\top 
\BK{\Bk{\vec{\check\beta}(x)-\vec{\beta}}\Bk{\vec{\check\beta}(x)-\vec{\beta}}^\top M(dx)}\mathbf{u}
&\ge (\mathbf{u}^\top \vec{X})^2.
\label{ineq}
\end{align}
Since $(\mathbf{u}^\top \vec{X})^2$ is Hermitian,
\begin{align}
(\mathbf{u}^\top\vec{X})^2 &= 
\sum_{s,t} u_s u_t \vec{X}_s \vec{X}_t = 
\bk{\sum_{s,t} u_s u_t \vec{X}_s \vec{X}_t}^\dagger
\nonumber\\
&= \sum_{s,t} u_s u_t \vec{X}_t \vec{X}_s  = 
\sum_{s,t} u_s u_t \vec{X}_s \circ \vec{X}_t = \mathbf{u}^\top(\vec{X}\circ \vec{X}^\top) \mathbf{u}.
\end{align}
Multiplying both sides of Eq.~(\ref{ineq}) by $\rho_{\vec{\theta}}$ and taking the
operator trace leads to the lemma.
\end{proof}

\subsection{Proof of Theorem~\ref{lem_nonemptyInfluence}}
\label{app:lem_nonemptyInfluence}
\begin{proof}
We prove the forward implication first.
We assume that there exists a $q$-dimensional vector $\vec{X}$ of influence operators  satisfying the local unbiasedness conditions
\begin{align}\label{eq:constr}
\trace \rho_{\vec{\theta}} L_j \circ X_s = \bk{\partial_{\vec{\theta}}\vec{\beta}}_{js}.
\end{align}
Let $\mathbf{u} \in \mathbb{R}^p$ and $\mathbf{v} \in \mathbb{R}^q$ be two arbitrary vectors.
From~\eqref{eq:constr} we have the scalar equality
\begin{equation}
\mathbf{v}^\top \left( \trace \rho_{\vec{\theta}} \vec{L} \circ \vec{X}^\top \right) \mathbf{u} = 
\trace \rho_{\vec{\theta}} \left( \mathbf{v}^\top  \vec{L} \right) \circ \left( \vec{X}^\top \mathbf{u} \right) = \mathbf{v}^\top ( \partial_{\vec{\theta}} \vec{\beta} ) \mathbf{u} .
\end{equation}
Introducing operators $V = \mathbf{v}^\top \vec{L}$ and $U = \mathbf{u}^\top \vec{X}$ we can use the Cauchy-Schwarz inequality for the SLD inner product:
\begin{equation}
| \trace \rho_{\vec{\theta}} V \circ U|^2 \leq \bk{\trace \rho_{\vec{\theta}} V^2}  \bk{\trace \rho_{\vec{\theta}} U^2},
\end{equation}
which can be rewritten as
\begin{align}
\left[ \mathbf{v}^\top \bk{\partial_{\vec{\theta}} \vec{\beta}} \mathbf{u} \right]^2 \leq & \left[ \mathbf{v}^\top \bk{\trace \rho_{\vec{\theta}} \vec{L} \circ \vec{L}^\top} \mathbf{v} \right] \left[ \mathbf{u}^\top \bk{\trace \rho_{\vec{\theta}} \vec{X} \circ \vec{X}^\top} \mathbf{u} \right] \\
& = \left( \mathbf{v}^\top J \mathbf{v} \right) \left[ \mathbf{u}^\top V(\vec{X}) \mathbf{u} \right] \notag
\end{align}
Now let $\mathbf{k}$ be a vector in the kernel of $J$ (such that $J \mathbf{k} = 0$) and choose $\mathbf{v}=\mathbf{k}$ and $\mathbf{u} = (\partial_{\vec{\theta}} \vec{\beta})^\top \mathbf{k}$ to get
\begin{equation}
0 \leq \left[ \mathbf{k}^\top (\partial_{\vec{\theta}} \vec{\beta})(\partial_{\vec{\theta}} \vec{\beta})^\top \mathbf{k} \right]^2 \leq 0,
\end{equation}
which implies $(\partial_{\vec{\theta}} \vec{\beta})^\top \mathbf{k}=0$.
Since this holds for any vector in the kernel of $J$, the columns of $\partial_{\vec{\theta}} \vec{\beta}$ must belong to the range of $J$, i.e. the orthogonal complement of the kernel.
In turn this implies $J J^{+} \partial \beta = \partial \beta$, since $J J^{+}$ is the orthogonal projector on the range of $J$~\cite{Banerjee1973}.

The reverse implication follows from the fact that $\vec{X}_0 \equiv (\partial_{\vec{\theta}} \vec{\beta})^\top J^{+} \vec{L} \in \mathcal{X}_{\vec{\theta}}$ when $J J^{+} \partial_{\vec{\theta}} \vec{\beta} = \partial_{\vec{\theta}} \vec{\beta}$, since
\begin{equation}
\Tr \rho_{\vec{\theta}} \vec{L} \circ \vec{X}_0^\top = \bk{\Tr \rho_{\vec{\theta}} \vec{L} \circ \vec{L}^\top} J^{+} \partial_{\vec{\theta}} \vec{\beta} = J J^{+} \partial_{\vec{\theta}} \vec{\beta}.
\end{equation}
\end{proof}

\subsection{Proof of Corollary~\ref{cor_hb}}
\label{app:cor_hb}
\begin{proof}
  $J\equiv V(\vec{L})$ is the Gram matrix of $\{\vec{L}\}$ with respect to the SLD inner product~\eqref{inner}.
  If $J$ is full-rank, $\{\vec{L}\}$ is linearly independent, and the formula for the projection $\Pi(\vec{X}|\mathcal T)$ is given by $\bk{\partial_{\vec{\theta}}\vec{\beta}}^\top J^{-1}\vec{L}$~\cite{tsiatis06,bickel93}.
  If $J$ is not full-rank and $\{\vec{L}\}$ is linearly dependent the same formula can be generalized by taking the Moore-Penrose pseudoinverse as follows.

  We introduce the linear map $\Lambda : \mathbbm{R}^p \mapsto \mathcal{L}_{\mathsf{h}}(\mathsf{H})$ that maps each element $e_i$ of the canonical basis of $\mathbbm{R}^p$ to the SLD operator $L_i$, i.e.
  \begin{equation} 
  \Lambda ( \vec{x} ) = \vec{L}^\T \vec{x};
  \end{equation}
  clearly, $\spn\{ \vec{L} \}=\cspn\{ \vec{L} \}$ is the image of $\Lambda$.
  The projector on this finite-dimensional subspace is given by $\Lambda \Lambda^{+}$, where $\Lambda^{+}$ denotes the Moore-Penrose pseudoinverse of the linear map~\cite[p.~163]{luenberger1997optimization}.
  In particular we can rewrite the pseudoinverse as $\Lambda^{+}=(\Lambda^\dag \Lambda)^{+} \Lambda^\dag$, where $\Lambda^\dag : \mathcal{L}_{\mathsf{h}}(\mathsf{H}) \mapsto \mathbbm{R}^p$ denotes the adjoint map w.r.t. the SLD inner product
  \begin{equation}
	\Lambda^\dag (X) = ( \langle X , L_1 \rangle, \dots, \langle X , L_p \rangle )^\top = \Tr \rho_{\vec{\theta}} X \circ \vec{L} .
  \end{equation}
  Since $\Lambda^\dag \Lambda=V(\vec{L})=J$ is the Gram matrix of the SLDs, the projection of a single operator $X_s$ is
  \begin{equation}
  \Pi(X_s|\mathcal{T}) = L_i (J^{+})_{ij}  \langle X_s , L_j \rangle.
  \end{equation}
  By applying the projector to all the elements of the replicating space we can write the vectoral form 
  \begin{equation}
  \Pi( \vec{\vec{X}}|\mathcal{T}) =  ( \Tr \rho_{\vec{\theta}} \vec{\vec{X}} \circ \vec{L}^\top ) J^{+} \vec{L} = (\partial_{\vec{\theta}} \vec{\beta})^\top J^{+} \vec{L};
  \end{equation}
  where we inserted the local unbiasedness conditions~\eqref{eq:locunbvec} to recover~\eqref{proj}, from which~\eqref{Veff} and~\eqref{CS} follow immediately. 
\end{proof}

\subsection{Proof of Lemma~\ref{lem_Z}}
\label{app:lem_Z}
\begin{proof}
  Let $f(x)$ be a complex function and $F = \int f(x) M(dx)$. Consider
  the inequality
\begin{align}
\int \Bk{f(x)-F}M(dx)\Bk{f^*(x)-F^\dagger} &\succeq 0,
\end{align}
which leads to 
\begin{align}
\int \abs{f(x)}^2 M(dx) 
\succeq  F F^\dagger.
\label{ineq2}
\end{align}
Now let $\mathbf{z}$ be an arbitrary complex vector and
\begin{align}
f(x) &= \sum_s z_s^*\Bk{\check\beta_s(x)-\beta_s},
&
F &= \sum_s z_s^* X_s.
\end{align}
Equation~(\ref{ineq2}) leads to
\begin{align}
\mathbf{z}^\dagger
\BK{\int \Bk{\vec{\check\beta}(x)-\vec{\beta}} 
\Bk{\vec{\check\beta}(x)- \vec{\beta}}^\top M(dx)}
\mathbf{z}
&\ge 
\mathbf{z}^\dagger \vec{X} \vec{X}^\top \mathbf{z}.
\end{align}
Multiplying both sides by $\rho_{\vec{\theta}}$ and taking the operator trace
leads to the lemma.
\end{proof}

\subsection{Proof of Lemma~\ref{lem_bg} (Belavkin-Grishanin inequality)}
\label{app:lem_bg}
\begin{proof}
  It is easy to prove that $A^* = \real A - i\imag A \succeq 0$,
  $\real A = (A+A^*)/2 \succeq 0$, $\real A$ and $i\imag A$ are Hermitian,
  and $\imag A$ is skew-symmetric.  Then, for any complex vector $\mathbf{z}$,
\begin{align}
\real A &\succeq \pm i\imag A,
&
\mathbf{z}^\dagger\bk{\real A} \mathbf{z} &\ge \abs{\mathbf{z}^\dagger \bk{i\imag A} \mathbf{z}}.
\end{align}
Let $\{\lambda_s\}$ be the eigenvalues of $i\imag A$
and $\{\mathbf{z}_s\}$ be the corresponding unit eigenvectors.
Then
\begin{align}
\trace \real A &= \sum_s \mathbf{z}_s^\dagger \bk{\real A} \mathbf{z}_s
\ge \sum_s \abs{ \mathbf{z}_s^\dagger \bk{i\imag A} \mathbf{z}_s}
\nonumber\\
&= \sum_s \abs{\lambda_s} = \norm{i\imag A}_1 = \norm{\imag A}_1.
\end{align}
\end{proof}

\section{Added note on a related work}\label{app:Carollo}
While completing this work we became aware of the recently published work of Carollo et. al~\cite{Carollo2019}, where the authors derive the inequality $C^\mathrm{H} \leq 2 C^\mathrm{S}$ using a different intermediate bound.
While the final result is correct, the paper~\cite{Carollo2019} contains a small mistake that makes the proof incomplete for generic weight matrices $W$; here we show how to fix this.

The derivation of~\cite{Carollo2019} starts from a wrong expression of the second term appearing in~\eqref{hcrb}: $\Vert W \Im Z(\vec{X}) \Vert_1$ instead of $\Vert \sqrt{W} \Im Z(\vec{X})  \sqrt{W} \Vert_1$\footnote{
This misunderstanding appeared also in some previous works~\cite{Genoni2013b,Ragy2019}, but since the matrix $W=\id$ is then considered the results are valid (it is enough that $W$ commutes with $A$).}.
However, the derived upper bound is still valid, since for a skew-symmetric $A$ we have the following inequality
\begin{equation} \label{eq:tracenormIneq}
\left\Vert \sqrt{W} A \sqrt{W} \right\Vert_1 \leq \left\Vert W A \right\Vert_1 .
\end{equation}
This comes from the fact that the matrices $M_1 = W A$ and $M_2 = \sqrt{W} A \sqrt{W}$ are isospectral, but only $M_2$ is normal (being skew-symmetric).
For a normal matrix the singular values are equal to the absolute values of the eigenvalues.
However, for a non-normal matrix such as $M_1$ the ordered vector containing the absolute values of the eigenvalues is majorized by the ordered vector of singular values~\cite{Bhatia1997}; in turn, this fact implies~\eqref{eq:tracenormIneq}.

\bibliography{HCRBvsSLDCRB}

\begin{thebibliography}{66}%
\makeatletter
\providecommand \@ifxundefined [1]{%
 \@ifx{#1\undefined}
}%
\providecommand \@ifnum [1]{%
 \ifnum #1\expandafter \@firstoftwo
 \else \expandafter \@secondoftwo
 \fi
}%
\providecommand \@ifx [1]{%
 \ifx #1\expandafter \@firstoftwo
 \else \expandafter \@secondoftwo
 \fi
}%
\providecommand \natexlab [1]{#1}%
\providecommand \enquote  [1]{``#1''}%
\providecommand \bibnamefont  [1]{#1}%
\providecommand \bibfnamefont [1]{#1}%
\providecommand \citenamefont [1]{#1}%
\providecommand \href@noop [0]{\@secondoftwo}%
\providecommand \href [0]{\begingroup \@sanitize@url \@href}%
\providecommand \@href[1]{\@@startlink{#1}\@@href}%
\providecommand \@@href[1]{\endgroup#1\@@endlink}%
\providecommand \@sanitize@url [0]{\catcode `\\12\catcode `\$12\catcode
  `\&12\catcode `\#12\catcode `\^12\catcode `\_12\catcode `\%12\relax}%
\providecommand \@@startlink[1]{}%
\providecommand \@@endlink[0]{}%
\providecommand \url  [0]{\begingroup\@sanitize@url \@url }%
\providecommand \@url [1]{\endgroup\@href {#1}{\urlprefix }}%
\providecommand \urlprefix  [0]{URL }%
\providecommand \Eprint [0]{\href }%
\providecommand \doibase [0]{http://dx.doi.org/}%
\providecommand \selectlanguage [0]{\@gobble}%
\providecommand \bibinfo  [0]{\@secondoftwo}%
\providecommand \bibfield  [0]{\@secondoftwo}%
\providecommand \translation [1]{[#1]}%
\providecommand \BibitemOpen [0]{}%
\providecommand \bibitemStop [0]{}%
\providecommand \bibitemNoStop [0]{.\EOS\space}%
\providecommand \EOS [0]{\spacefactor3000\relax}%
\providecommand \BibitemShut  [1]{\csname bibitem#1\endcsname}%
\let\auto@bib@innerbib\@empty
\bibitem [{\citenamefont {Humphreys}\ \emph {et~al.}(2013)\citenamefont
  {Humphreys}, \citenamefont {Barbieri}, \citenamefont {Datta},\ and\
  \citenamefont {Walmsley}}]{Humphreys2013}%
  \BibitemOpen
  \bibfield  {author} {\bibinfo {author} {\bibfnamefont {P.~C.}\ \bibnamefont
  {Humphreys}}, \bibinfo {author} {\bibfnamefont {M.}~\bibnamefont {Barbieri}},
  \bibinfo {author} {\bibfnamefont {A.}~\bibnamefont {Datta}}, \ and\ \bibinfo
  {author} {\bibfnamefont {I.~A.}\ \bibnamefont {Walmsley}},\ }\href {\doibase
  10.1103/PhysRevLett.111.070403} {\bibfield  {journal} {\bibinfo  {journal}
  {Phys. Rev. Lett.}\ }\textbf {\bibinfo {volume} {111}},\ \bibinfo {pages}
  {070403} (\bibinfo {year} {2013})}\BibitemShut {NoStop}%
\bibitem [{\citenamefont {Vidrighin}\ \emph {et~al.}(2014)\citenamefont
  {Vidrighin}, \citenamefont {Donati}, \citenamefont {Genoni}, \citenamefont
  {Jin}, \citenamefont {Kolthammer}, \citenamefont {Kim}, \citenamefont
  {Datta}, \citenamefont {Barbieri},\ and\ \citenamefont
  {Walmsley}}]{Vidrighin2014}%
  \BibitemOpen
  \bibfield  {author} {\bibinfo {author} {\bibfnamefont {M.~D.}\ \bibnamefont
  {Vidrighin}}, \bibinfo {author} {\bibfnamefont {G.}~\bibnamefont {Donati}},
  \bibinfo {author} {\bibfnamefont {M.~G.}\ \bibnamefont {Genoni}}, \bibinfo
  {author} {\bibfnamefont {X.-M.}\ \bibnamefont {Jin}}, \bibinfo {author}
  {\bibfnamefont {W.~S.}\ \bibnamefont {Kolthammer}}, \bibinfo {author}
  {\bibfnamefont {M.~S.}\ \bibnamefont {Kim}}, \bibinfo {author} {\bibfnamefont
  {A.}~\bibnamefont {Datta}}, \bibinfo {author} {\bibfnamefont
  {M.}~\bibnamefont {Barbieri}}, \ and\ \bibinfo {author} {\bibfnamefont
  {I.~A.}\ \bibnamefont {Walmsley}},\ }\href {\doibase 10.1038/ncomms4532}
  {\bibfield  {journal} {\bibinfo  {journal} {Nat. Commun.}\ }\textbf {\bibinfo
  {volume} {5}},\ \bibinfo {pages} {3532} (\bibinfo {year} {2014})}\BibitemShut
  {NoStop}%
\bibitem [{\citenamefont {Baumgratz}\ and\ \citenamefont
  {Datta}(2016)}]{Baumgratz2015}%
  \BibitemOpen
  \bibfield  {author} {\bibinfo {author} {\bibfnamefont {T.}~\bibnamefont
  {Baumgratz}}\ and\ \bibinfo {author} {\bibfnamefont {A.}~\bibnamefont
  {Datta}},\ }\href {\doibase 10.1103/PhysRevLett.116.030801} {\bibfield
  {journal} {\bibinfo  {journal} {Phys. Rev. Lett.}\ }\textbf {\bibinfo
  {volume} {116}},\ \bibinfo {pages} {030801} (\bibinfo {year}
  {2016})}\BibitemShut {NoStop}%
\bibitem [{\citenamefont {Szczykulska}\ \emph {et~al.}(2016)\citenamefont
  {Szczykulska}, \citenamefont {Baumgratz},\ and\ \citenamefont
  {Datta}}]{Szczykulska2016}%
  \BibitemOpen
  \bibfield  {author} {\bibinfo {author} {\bibfnamefont {M.}~\bibnamefont
  {Szczykulska}}, \bibinfo {author} {\bibfnamefont {T.}~\bibnamefont
  {Baumgratz}}, \ and\ \bibinfo {author} {\bibfnamefont {A.}~\bibnamefont
  {Datta}},\ }\href {\doibase 10.1080/23746149.2016.1230476} {\bibfield
  {journal} {\bibinfo  {journal} {Adv. Phys. X}\ }\textbf {\bibinfo {volume}
  {1}},\ \bibinfo {pages} {621} (\bibinfo {year} {2016})}\BibitemShut {NoStop}%
\bibitem [{\citenamefont {Gagatsos}\ \emph {et~al.}(2016)\citenamefont
  {Gagatsos}, \citenamefont {Branford},\ and\ \citenamefont
  {Datta}}]{Gagatsos2016a}%
  \BibitemOpen
  \bibfield  {author} {\bibinfo {author} {\bibfnamefont {C.~N.}\ \bibnamefont
  {Gagatsos}}, \bibinfo {author} {\bibfnamefont {D.}~\bibnamefont {Branford}},
  \ and\ \bibinfo {author} {\bibfnamefont {A.}~\bibnamefont {Datta}},\ }\href
  {\doibase 10.1103/PhysRevA.94.042342} {\bibfield  {journal} {\bibinfo
  {journal} {Phys. Rev. A}\ }\textbf {\bibinfo {volume} {94}},\ \bibinfo
  {pages} {042342} (\bibinfo {year} {2016})}\BibitemShut {NoStop}%
\bibitem [{\citenamefont {Chrostowski}\ \emph {et~al.}(2017)\citenamefont
  {Chrostowski}, \citenamefont {Demkowicz-Dobrza{\'{n}}ski}, \citenamefont
  {Jarzyna},\ and\ \citenamefont {Banaszek}}]{Chrostowski2017}%
  \BibitemOpen
  \bibfield  {author} {\bibinfo {author} {\bibfnamefont {A.}~\bibnamefont
  {Chrostowski}}, \bibinfo {author} {\bibfnamefont {R.}~\bibnamefont
  {Demkowicz-Dobrza{\'{n}}ski}}, \bibinfo {author} {\bibfnamefont
  {M.}~\bibnamefont {Jarzyna}}, \ and\ \bibinfo {author} {\bibfnamefont
  {K.}~\bibnamefont {Banaszek}},\ }\href {\doibase 10.1142/S0219749917400056}
  {\bibfield  {journal} {\bibinfo  {journal} {Int. J. Quantum Inf.}\ }\textbf
  {\bibinfo {volume} {15}},\ \bibinfo {pages} {1740005} (\bibinfo {year}
  {2017})}\BibitemShut {NoStop}%
\bibitem [{\citenamefont {Pezz{\`{e}}}\ \emph {et~al.}(2017)\citenamefont
  {Pezz{\`{e}}}, \citenamefont {Ciampini}, \citenamefont {Spagnolo},
  \citenamefont {Humphreys}, \citenamefont {Datta}, \citenamefont {Walmsley},
  \citenamefont {Barbieri}, \citenamefont {Sciarrino},\ and\ \citenamefont
  {Smerzi}}]{Pezze2017}%
  \BibitemOpen
  \bibfield  {author} {\bibinfo {author} {\bibfnamefont {L.}~\bibnamefont
  {Pezz{\`{e}}}}, \bibinfo {author} {\bibfnamefont {M.~A.}\ \bibnamefont
  {Ciampini}}, \bibinfo {author} {\bibfnamefont {N.}~\bibnamefont {Spagnolo}},
  \bibinfo {author} {\bibfnamefont {P.~C.}\ \bibnamefont {Humphreys}}, \bibinfo
  {author} {\bibfnamefont {A.}~\bibnamefont {Datta}}, \bibinfo {author}
  {\bibfnamefont {I.~A.}\ \bibnamefont {Walmsley}}, \bibinfo {author}
  {\bibfnamefont {M.}~\bibnamefont {Barbieri}}, \bibinfo {author}
  {\bibfnamefont {F.}~\bibnamefont {Sciarrino}}, \ and\ \bibinfo {author}
  {\bibfnamefont {A.}~\bibnamefont {Smerzi}},\ }\href {\doibase
  10.1103/PhysRevLett.119.130504} {\bibfield  {journal} {\bibinfo  {journal}
  {Phys. Rev. Lett.}\ }\textbf {\bibinfo {volume} {119}},\ \bibinfo {pages}
  {130504} (\bibinfo {year} {2017})}\BibitemShut {NoStop}%
\bibitem [{\citenamefont {Roccia}\ \emph
  {et~al.}(2018{\natexlab{a}})\citenamefont {Roccia}, \citenamefont {Gianani},
  \citenamefont {Mancino}, \citenamefont {Sbroscia}, \citenamefont {Somma},
  \citenamefont {Genoni},\ and\ \citenamefont {Barbieri}}]{Roccia2017}%
  \BibitemOpen
  \bibfield  {author} {\bibinfo {author} {\bibfnamefont {E.}~\bibnamefont
  {Roccia}}, \bibinfo {author} {\bibfnamefont {I.}~\bibnamefont {Gianani}},
  \bibinfo {author} {\bibfnamefont {L.}~\bibnamefont {Mancino}}, \bibinfo
  {author} {\bibfnamefont {M.}~\bibnamefont {Sbroscia}}, \bibinfo {author}
  {\bibfnamefont {F.}~\bibnamefont {Somma}}, \bibinfo {author} {\bibfnamefont
  {M.~G.}\ \bibnamefont {Genoni}}, \ and\ \bibinfo {author} {\bibfnamefont
  {M.}~\bibnamefont {Barbieri}},\ }\href {\doibase 10.1088/2058-9565/aa9212}
  {\bibfield  {journal} {\bibinfo  {journal} {Quantum Sci. Technol.}\ }\textbf
  {\bibinfo {volume} {3}},\ \bibinfo {pages} {01LT01} (\bibinfo {year}
  {2018}{\natexlab{a}})}\BibitemShut {NoStop}%
\bibitem [{\citenamefont {Roccia}\ \emph
  {et~al.}(2018{\natexlab{b}})\citenamefont {Roccia}, \citenamefont {Cimini},
  \citenamefont {Sbroscia}, \citenamefont {Gianani}, \citenamefont {Ruggiero},
  \citenamefont {Mancino}, \citenamefont {Genoni}, \citenamefont {Ricci},\ and\
  \citenamefont {Barbieri}}]{Roccia2018}%
  \BibitemOpen
  \bibfield  {author} {\bibinfo {author} {\bibfnamefont {E.}~\bibnamefont
  {Roccia}}, \bibinfo {author} {\bibfnamefont {V.}~\bibnamefont {Cimini}},
  \bibinfo {author} {\bibfnamefont {M.}~\bibnamefont {Sbroscia}}, \bibinfo
  {author} {\bibfnamefont {I.}~\bibnamefont {Gianani}}, \bibinfo {author}
  {\bibfnamefont {L.}~\bibnamefont {Ruggiero}}, \bibinfo {author}
  {\bibfnamefont {L.}~\bibnamefont {Mancino}}, \bibinfo {author} {\bibfnamefont
  {M.~G.}\ \bibnamefont {Genoni}}, \bibinfo {author} {\bibfnamefont {M.~A.}\
  \bibnamefont {Ricci}}, \ and\ \bibinfo {author} {\bibfnamefont
  {M.}~\bibnamefont {Barbieri}},\ }\href {\doibase 10.1364/OPTICA.5.001171}
  {\bibfield  {journal} {\bibinfo  {journal} {Optica}\ }\textbf {\bibinfo
  {volume} {5}},\ \bibinfo {pages} {1171} (\bibinfo {year}
  {2018}{\natexlab{b}})}\BibitemShut {NoStop}%
\bibitem [{\citenamefont {Gessner}\ \emph {et~al.}(2018)\citenamefont
  {Gessner}, \citenamefont {Pezz{\`{e}}},\ and\ \citenamefont
  {Smerzi}}]{Gessner2018}%
  \BibitemOpen
  \bibfield  {author} {\bibinfo {author} {\bibfnamefont {M.}~\bibnamefont
  {Gessner}}, \bibinfo {author} {\bibfnamefont {L.}~\bibnamefont
  {Pezz{\`{e}}}}, \ and\ \bibinfo {author} {\bibfnamefont {A.}~\bibnamefont
  {Smerzi}},\ }\href {\doibase 10.1103/PhysRevLett.121.130503} {\bibfield
  {journal} {\bibinfo  {journal} {Phys. Rev. Lett.}\ }\textbf {\bibinfo
  {volume} {121}},\ \bibinfo {pages} {130503} (\bibinfo {year}
  {2018})}\BibitemShut {NoStop}%
\bibitem [{\citenamefont {Yang}\ \emph
  {et~al.}(2019{\natexlab{a}})\citenamefont {Yang}, \citenamefont {Pang},
  \citenamefont {Zhou},\ and\ \citenamefont {Jordan}}]{Yang2018b}%
  \BibitemOpen
  \bibfield  {author} {\bibinfo {author} {\bibfnamefont {J.}~\bibnamefont
  {Yang}}, \bibinfo {author} {\bibfnamefont {S.}~\bibnamefont {Pang}}, \bibinfo
  {author} {\bibfnamefont {Y.}~\bibnamefont {Zhou}}, \ and\ \bibinfo {author}
  {\bibfnamefont {A.~N.}\ \bibnamefont {Jordan}},\ }\href {\doibase
  10.1103/PhysRevA.100.032104} {\bibfield  {journal} {\bibinfo  {journal}
  {Phys. Rev. A}\ }\textbf {\bibinfo {volume} {100}},\ \bibinfo {pages}
  {032104} (\bibinfo {year} {2019}{\natexlab{a}})}\BibitemShut {NoStop}%
\bibitem [{\citenamefont {Genoni}\ and\ \citenamefont
  {Tufarelli}(2019)}]{Genoni2019}%
  \BibitemOpen
  \bibfield  {author} {\bibinfo {author} {\bibfnamefont {M.~G.}\ \bibnamefont
  {Genoni}}\ and\ \bibinfo {author} {\bibfnamefont {T.}~\bibnamefont
  {Tufarelli}},\ }\href {\doibase 10.1088/1751-8121/ab3fe0} {\bibfield
  {journal} {\bibinfo  {journal} {J. Phys. A Math. Theor.}\ }\textbf {\bibinfo
  {volume} {52}},\ \bibinfo {pages} {434002} (\bibinfo {year}
  {2019})}\BibitemShut {NoStop}%
\bibitem [{\citenamefont {Gorecki}\ \emph {et~al.}(2019)\citenamefont
  {Gorecki}, \citenamefont {Zhou}, \citenamefont {Jiang},\ and\ \citenamefont
  {Demkowicz-Dobrza{\'{n}}ski}}]{Gorecki2019}%
  \BibitemOpen
  \bibfield  {author} {\bibinfo {author} {\bibfnamefont {W.}~\bibnamefont
  {Gorecki}}, \bibinfo {author} {\bibfnamefont {S.}~\bibnamefont {Zhou}},
  \bibinfo {author} {\bibfnamefont {L.}~\bibnamefont {Jiang}}, \ and\ \bibinfo
  {author} {\bibfnamefont {R.}~\bibnamefont {Demkowicz-Dobrza{\'{n}}ski}},\
  }\href {https://arxiv.org/abs/1901.00896} {\bibfield  {journal} {\bibinfo
  {journal} {arXiv:1901.00896}\ } (\bibinfo {year} {2019})}\BibitemShut
  {NoStop}%
\bibitem [{\citenamefont {Rubio}\ and\ \citenamefont
  {Dunningham}(2019)}]{Rubio2019}%
  \BibitemOpen
  \bibfield  {author} {\bibinfo {author} {\bibfnamefont {J.}~\bibnamefont
  {Rubio}}\ and\ \bibinfo {author} {\bibfnamefont {J.}~\bibnamefont
  {Dunningham}},\ }\href {http://arxiv.org/abs/1906.04123} {\bibfield
  {journal} {\bibinfo  {journal} {arXiv:1906.04123}\ } (\bibinfo {year}
  {2019})}\BibitemShut {NoStop}%
\bibitem [{\citenamefont {Bisketzi}\ \emph {et~al.}(2019)\citenamefont
  {Bisketzi}, \citenamefont {Branford},\ and\ \citenamefont
  {Datta}}]{Bisketzi2019}%
  \BibitemOpen
  \bibfield  {author} {\bibinfo {author} {\bibfnamefont {E.}~\bibnamefont
  {Bisketzi}}, \bibinfo {author} {\bibfnamefont {D.}~\bibnamefont {Branford}},
  \ and\ \bibinfo {author} {\bibfnamefont {A.}~\bibnamefont {Datta}},\ }\href
  {\doibase 10.1088/1367-2630/ab58a0} {\bibfield  {journal} {\bibinfo
  {journal} {New J. Phys.}\ }\textbf {\bibinfo {volume} {21}},\ \bibinfo
  {pages} {123032} (\bibinfo {year} {2019})}\BibitemShut {NoStop}%
\bibitem [{\citenamefont {Liu}\ \emph {et~al.}(2019)\citenamefont {Liu},
  \citenamefont {Yuan}, \citenamefont {Lu},\ and\ \citenamefont
  {Wang}}]{Liu2019d}%
  \BibitemOpen
  \bibfield  {author} {\bibinfo {author} {\bibfnamefont {J.}~\bibnamefont
  {Liu}}, \bibinfo {author} {\bibfnamefont {H.}~\bibnamefont {Yuan}}, \bibinfo
  {author} {\bibfnamefont {X.-m.}\ \bibnamefont {Lu}}, \ and\ \bibinfo {author}
  {\bibfnamefont {X.}~\bibnamefont {Wang}},\ }\href
  {http://arxiv.org/abs/1907.08037} {\bibfield  {journal} {\bibinfo  {journal}
  {arXiv:1907.08037}\ } (\bibinfo {year} {2019})}\BibitemShut {NoStop}%
\bibitem [{\citenamefont {Lu}\ \emph {et~al.}(2019)\citenamefont {Lu},
  \citenamefont {Ma},\ and\ \citenamefont {Zhang}}]{Lu2019}%
  \BibitemOpen
  \bibfield  {author} {\bibinfo {author} {\bibfnamefont {X.-M.}\ \bibnamefont
  {Lu}}, \bibinfo {author} {\bibfnamefont {Z.}~\bibnamefont {Ma}}, \ and\
  \bibinfo {author} {\bibfnamefont {C.}~\bibnamefont {Zhang}},\ }\href
  {http://arxiv.org/abs/1910.06035} {\bibfield  {journal} {\bibinfo  {journal}
  {arXiv:1910.06035}\ } (\bibinfo {year} {2019})}\BibitemShut {NoStop}%
\bibitem [{\citenamefont {Albarelli}\ \emph
  {et~al.}(2019{\natexlab{a}})\citenamefont {Albarelli}, \citenamefont
  {Barbieri}, \citenamefont {Genoni},\ and\ \citenamefont
  {Gianani}}]{Albarelli2019c}%
  \BibitemOpen
  \bibfield  {author} {\bibinfo {author} {\bibfnamefont {F.}~\bibnamefont
  {Albarelli}}, \bibinfo {author} {\bibfnamefont {M.}~\bibnamefont {Barbieri}},
  \bibinfo {author} {\bibfnamefont {M.~G.}\ \bibnamefont {Genoni}}, \ and\
  \bibinfo {author} {\bibfnamefont {I.}~\bibnamefont {Gianani}},\ }\href
  {http://arxiv.org/abs/1911.12067} {\bibfield  {journal} {\bibinfo  {journal}
  {arXiv:1911.12067}\ } (\bibinfo {year} {2019}{\natexlab{a}})}\BibitemShut
  {NoStop}%
\bibitem [{\citenamefont {Helstrom}(1967)}]{Helstrom1967}%
  \BibitemOpen
  \bibfield  {author} {\bibinfo {author} {\bibfnamefont {C.~W.}\ \bibnamefont
  {Helstrom}},\ }\href {\doibase 10.1016/0375-9601(67)90366-0} {\bibfield
  {journal} {\bibinfo  {journal} {Phys. Lett. A}\ }\textbf {\bibinfo {volume}
  {25}},\ \bibinfo {pages} {101} (\bibinfo {year} {1967})}\BibitemShut
  {NoStop}%
\bibitem [{\citenamefont {Helstrom}(1976)}]{helstrom1976quantum}%
  \BibitemOpen
  \bibfield  {author} {\bibinfo {author} {\bibfnamefont {C.~W.}\ \bibnamefont
  {Helstrom}},\ }\href@noop {} {\emph {\bibinfo {title} {{Quantum detection and
  estimation theory}}}}\ (\bibinfo  {publisher} {Academic Press},\ \bibinfo
  {address} {New York},\ \bibinfo {year} {1976})\BibitemShut {NoStop}%
\bibitem [{\citenamefont {Yuen}\ and\ \citenamefont {Lax}(1973)}]{Yuen1973}%
  \BibitemOpen
  \bibfield  {author} {\bibinfo {author} {\bibfnamefont {H.~P.}\ \bibnamefont
  {Yuen}}\ and\ \bibinfo {author} {\bibfnamefont {M.}~\bibnamefont {Lax}},\
  }\href {\doibase 10.1109/TIT.1973.1055103} {\bibfield  {journal} {\bibinfo
  {journal} {IEEE Trans. Inf. Theory}\ }\textbf {\bibinfo {volume} {19}},\
  \bibinfo {pages} {740} (\bibinfo {year} {1973})}\BibitemShut {NoStop}%
\bibitem [{\citenamefont {Belavkin}(1976)}]{Belavkin1976}%
  \BibitemOpen
  \bibfield  {author} {\bibinfo {author} {\bibfnamefont {V.~P.}\ \bibnamefont
  {Belavkin}},\ }\href {\doibase 10.1007/BF01032091} {\bibfield  {journal}
  {\bibinfo  {journal} {Theor. Math. Phys.}\ }\textbf {\bibinfo {volume}
  {26}},\ \bibinfo {pages} {213} (\bibinfo {year} {1976})}\BibitemShut
  {NoStop}%
\bibitem [{\citenamefont {Holevo}(1976)}]{Holevo1976}%
  \BibitemOpen
  \bibfield  {author} {\bibinfo {author} {\bibfnamefont {A.~S.}\ \bibnamefont
  {Holevo}},\ }in\ \href {\doibase 10.1007/BFb0077479} {\emph {\bibinfo
  {booktitle} {Proc. Third Japan — USSR Symp. Probab. Theory}}},\ \bibinfo
  {series} {Lecture Notes in Mathematics}, Vol.\ \bibinfo {volume} {550},\
  \bibinfo {editor} {edited by\ \bibinfo {editor} {\bibfnamefont
  {G.}~\bibnamefont {Maruyama}}\ and\ \bibinfo {editor} {\bibfnamefont {J.~V.}\
  \bibnamefont {Prokhorov}}}\ (\bibinfo  {publisher} {Springer},\ \bibinfo
  {address} {Berlin, Heidelberg},\ \bibinfo {year} {1976})\BibitemShut
  {NoStop}%
\bibitem [{\citenamefont {Holevo}(2011)}]{Holevo2011b}%
  \BibitemOpen
  \bibfield  {author} {\bibinfo {author} {\bibfnamefont {A.~S.}\ \bibnamefont
  {Holevo}},\ }\href {\doibase 10.1007/978-88-7642-378-9} {\emph {\bibinfo
  {title} {{Probabilistic and Statistical Aspects of Quantum Theory}}}},\
  \bibinfo {edition} {2nd}\ ed.\ (\bibinfo  {publisher} {Edizioni della
  Normale},\ \bibinfo {address} {Pisa},\ \bibinfo {year} {2011})\BibitemShut
  {NoStop}%
\bibitem [{\citenamefont {Nagaoka}(1989)}]{Nagaoka1989}%
  \BibitemOpen
  \bibfield  {author} {\bibinfo {author} {\bibfnamefont {H.}~\bibnamefont
  {Nagaoka}},\ }\href@noop {} {\bibfield  {journal} {\bibinfo  {journal} {IEICE
  Tech. Rep.}\ }\textbf {\bibinfo {volume} {IT 89-42}},\ \bibinfo {pages} {9}
  (\bibinfo {year} {1989})}\BibitemShut {NoStop}%
\bibitem [{\citenamefont {Hayashi}\ and\ \citenamefont
  {Matsumoto}(2008)}]{Hayashi2008a}%
  \BibitemOpen
  \bibfield  {author} {\bibinfo {author} {\bibfnamefont {M.}~\bibnamefont
  {Hayashi}}\ and\ \bibinfo {author} {\bibfnamefont {K.}~\bibnamefont
  {Matsumoto}},\ }\href {\doibase 10.1063/1.2988130} {\bibfield  {journal}
  {\bibinfo  {journal} {J. Math. Phys.}\ }\textbf {\bibinfo {volume} {49}},\
  \bibinfo {pages} {102101} (\bibinfo {year} {2008})}\BibitemShut {NoStop}%
\bibitem [{\citenamefont {Yamagata}\ \emph {et~al.}(2013)\citenamefont
  {Yamagata}, \citenamefont {Fujiwara},\ and\ \citenamefont
  {Gill}}]{Yamagata2013}%
  \BibitemOpen
  \bibfield  {author} {\bibinfo {author} {\bibfnamefont {K.}~\bibnamefont
  {Yamagata}}, \bibinfo {author} {\bibfnamefont {A.}~\bibnamefont {Fujiwara}},
  \ and\ \bibinfo {author} {\bibfnamefont {R.~D.}\ \bibnamefont {Gill}},\
  }\href {\doibase 10.1214/13-AOS1147} {\bibfield  {journal} {\bibinfo
  {journal} {Ann. Stat.}\ }\textbf {\bibinfo {volume} {41}},\ \bibinfo {pages}
  {2197} (\bibinfo {year} {2013})}\BibitemShut {NoStop}%
\bibitem [{\citenamefont {Yang}\ \emph
  {et~al.}(2019{\natexlab{b}})\citenamefont {Yang}, \citenamefont
  {Chiribella},\ and\ \citenamefont {Hayashi}}]{Yang2018a}%
  \BibitemOpen
  \bibfield  {author} {\bibinfo {author} {\bibfnamefont {Y.}~\bibnamefont
  {Yang}}, \bibinfo {author} {\bibfnamefont {G.}~\bibnamefont {Chiribella}}, \
  and\ \bibinfo {author} {\bibfnamefont {M.}~\bibnamefont {Hayashi}},\ }\href
  {\doibase 10.1007/s00220-019-03433-4} {\bibfield  {journal} {\bibinfo
  {journal} {Commun. Math. Phys.}\ }\textbf {\bibinfo {volume} {368}},\
  \bibinfo {pages} {223} (\bibinfo {year} {2019}{\natexlab{b}})}\BibitemShut
  {NoStop}%
\bibitem [{\citenamefont {Suzuki}(2016)}]{Suzuki2016a}%
  \BibitemOpen
  \bibfield  {author} {\bibinfo {author} {\bibfnamefont {J.}~\bibnamefont
  {Suzuki}},\ }\href {\doibase 10.1063/1.4945086} {\bibfield  {journal}
  {\bibinfo  {journal} {J. Math. Phys.}\ }\textbf {\bibinfo {volume} {57}},\
  \bibinfo {pages} {042201} (\bibinfo {year} {2016})}\BibitemShut {NoStop}%
\bibitem [{\citenamefont {Bradshaw}\ \emph {et~al.}(2017)\citenamefont
  {Bradshaw}, \citenamefont {Assad},\ and\ \citenamefont
  {Lam}}]{Bradshaw2017a}%
  \BibitemOpen
  \bibfield  {author} {\bibinfo {author} {\bibfnamefont {M.}~\bibnamefont
  {Bradshaw}}, \bibinfo {author} {\bibfnamefont {S.~M.}\ \bibnamefont {Assad}},
  \ and\ \bibinfo {author} {\bibfnamefont {P.~K.}\ \bibnamefont {Lam}},\ }\href
  {\doibase 10.1016/j.physleta.2017.06.024} {\bibfield  {journal} {\bibinfo
  {journal} {Phys. Lett. A}\ }\textbf {\bibinfo {volume} {381}},\ \bibinfo
  {pages} {2598} (\bibinfo {year} {2017})}\BibitemShut {NoStop}%
\bibitem [{\citenamefont {Bradshaw}\ \emph {et~al.}(2018)\citenamefont
  {Bradshaw}, \citenamefont {Lam},\ and\ \citenamefont {Assad}}]{Bradshaw2017}%
  \BibitemOpen
  \bibfield  {author} {\bibinfo {author} {\bibfnamefont {M.}~\bibnamefont
  {Bradshaw}}, \bibinfo {author} {\bibfnamefont {P.~K.}\ \bibnamefont {Lam}}, \
  and\ \bibinfo {author} {\bibfnamefont {S.~M.}\ \bibnamefont {Assad}},\ }\href
  {\doibase 10.1103/PhysRevA.97.012106} {\bibfield  {journal} {\bibinfo
  {journal} {Phys. Rev. A}\ }\textbf {\bibinfo {volume} {97}},\ \bibinfo
  {pages} {012106} (\bibinfo {year} {2018})}\BibitemShut {NoStop}%
\bibitem [{\citenamefont {Suzuki}(2019)}]{Suzuki2018}%
  \BibitemOpen
  \bibfield  {author} {\bibinfo {author} {\bibfnamefont {J.}~\bibnamefont
  {Suzuki}},\ }\href {\doibase 10.3390/e21070703} {\bibfield  {journal}
  {\bibinfo  {journal} {Entropy}\ }\textbf {\bibinfo {volume} {21}},\ \bibinfo
  {pages} {703} (\bibinfo {year} {2019})}\BibitemShut {NoStop}%
\bibitem [{\citenamefont {Sidhu}\ \emph {et~al.}(2019)\citenamefont {Sidhu},
  \citenamefont {Ouyang}, \citenamefont {Campbell},\ and\ \citenamefont
  {Kok}}]{Sidhu2019a}%
  \BibitemOpen
  \bibfield  {author} {\bibinfo {author} {\bibfnamefont {J.~S.}\ \bibnamefont
  {Sidhu}}, \bibinfo {author} {\bibfnamefont {Y.}~\bibnamefont {Ouyang}},
  \bibinfo {author} {\bibfnamefont {E.~T.}\ \bibnamefont {Campbell}}, \ and\
  \bibinfo {author} {\bibfnamefont {P.}~\bibnamefont {Kok}},\ }\href
  {http://arxiv.org/abs/1912.09218} {\bibfield  {journal} {\bibinfo  {journal}
  {arXiv:1912.09218}\ } (\bibinfo {year} {2019})}\BibitemShut {NoStop}%
\bibitem [{\citenamefont {Albarelli}\ \emph
  {et~al.}(2019{\natexlab{b}})\citenamefont {Albarelli}, \citenamefont
  {Friel},\ and\ \citenamefont {Datta}}]{Albarelli2019}%
  \BibitemOpen
  \bibfield  {author} {\bibinfo {author} {\bibfnamefont {F.}~\bibnamefont
  {Albarelli}}, \bibinfo {author} {\bibfnamefont {J.~F.}\ \bibnamefont
  {Friel}}, \ and\ \bibinfo {author} {\bibfnamefont {A.}~\bibnamefont
  {Datta}},\ }\href {\doibase 10.1103/PhysRevLett.123.200503} {\bibfield
  {journal} {\bibinfo  {journal} {Phys. Rev. Lett.}\ }\textbf {\bibinfo
  {volume} {123}},\ \bibinfo {pages} {200503} (\bibinfo {year}
  {2019}{\natexlab{b}})}\BibitemShut {NoStop}%
\bibitem [{\citenamefont {Suzuki}\ \emph {et~al.}(2019)\citenamefont {Suzuki},
  \citenamefont {Yang},\ and\ \citenamefont {Hayashi}}]{Suzuki2019a}%
  \BibitemOpen
  \bibfield  {author} {\bibinfo {author} {\bibfnamefont {J.}~\bibnamefont
  {Suzuki}}, \bibinfo {author} {\bibfnamefont {Y.}~\bibnamefont {Yang}}, \ and\
  \bibinfo {author} {\bibfnamefont {M.}~\bibnamefont {Hayashi}},\ }\href
  {http://arxiv.org/abs/1911.02790} {\bibfield  {journal} {\bibinfo  {journal}
  {arXiv:1911.02790}\ } (\bibinfo {year} {2019})}\BibitemShut {NoStop}%
\bibitem [{\citenamefont {Tsang}(2019)}]{Tsang2019}%
  \BibitemOpen
  \bibfield  {author} {\bibinfo {author} {\bibfnamefont {M.}~\bibnamefont
  {Tsang}},\ }\href {http://arxiv.org/abs/1906.09871} {\bibfield  {journal}
  {\bibinfo  {journal} {arXiv:1906.09871}\ } (\bibinfo {year}
  {2019})}\BibitemShut {NoStop}%
\bibitem [{\citenamefont {Li}\ \emph {et~al.}(2016)\citenamefont {Li},
  \citenamefont {Ferrie}, \citenamefont {Gross}, \citenamefont {Kalev},\ and\
  \citenamefont {Caves}}]{Li2016g}%
  \BibitemOpen
  \bibfield  {author} {\bibinfo {author} {\bibfnamefont {N.}~\bibnamefont
  {Li}}, \bibinfo {author} {\bibfnamefont {C.}~\bibnamefont {Ferrie}}, \bibinfo
  {author} {\bibfnamefont {J.~A.}\ \bibnamefont {Gross}}, \bibinfo {author}
  {\bibfnamefont {A.}~\bibnamefont {Kalev}}, \ and\ \bibinfo {author}
  {\bibfnamefont {C.~M.}\ \bibnamefont {Caves}},\ }\href {\doibase
  10.1103/PhysRevLett.116.180402} {\bibfield  {journal} {\bibinfo  {journal}
  {Phys. Rev. Lett.}\ }\textbf {\bibinfo {volume} {116}},\ \bibinfo {pages}
  {180402} (\bibinfo {year} {2016})}\BibitemShut {NoStop}%
\bibitem [{\citenamefont {Matsumoto}(2002)}]{Matsumoto2002}%
  \BibitemOpen
  \bibfield  {author} {\bibinfo {author} {\bibfnamefont {K.}~\bibnamefont
  {Matsumoto}},\ }\href {\doibase 10.1088/0305-4470/35/13/307} {\bibfield
  {journal} {\bibinfo  {journal} {J. Phys. A}\ }\textbf {\bibinfo {volume}
  {35}},\ \bibinfo {pages} {3111} (\bibinfo {year} {2002})}\BibitemShut
  {NoStop}%
\bibitem [{\citenamefont {Friel}\ \emph {et~al.}()\citenamefont {Friel},
  \citenamefont {Palittapongarnpim}, \citenamefont {Albarelli},\ and\
  \citenamefont {Datta}}]{inprep}%
  \BibitemOpen
  \bibfield  {author} {\bibinfo {author} {\bibfnamefont {J.~F.}\ \bibnamefont
  {Friel}}, \bibinfo {author} {\bibfnamefont {P.}~\bibnamefont
  {Palittapongarnpim}}, \bibinfo {author} {\bibfnamefont {F.}~\bibnamefont
  {Albarelli}}, \ and\ \bibinfo {author} {\bibfnamefont {A.}~\bibnamefont
  {Datta}},\ }\href@noop {} {\bibinfo  {journal} {{In preparation}}\
  }\BibitemShut {NoStop}%
\bibitem [{\citenamefont {Kahn}\ and\ \citenamefont {Guţă}(2009)}]{Kahn2009}%
  \BibitemOpen
\bibfield  {journal} {  }\bibfield  {author} {\bibinfo {author} {\bibfnamefont
  {J.}~\bibnamefont {Kahn}}\ and\ \bibinfo {author} {\bibfnamefont
  {M.}~\bibnamefont {Guţă}},\ }\href {\doibase 10.1007/s00220-009-0787-3}
  {\bibfield  {journal} {\bibinfo  {journal} {Commun. Math. Phys.}\ }\textbf
  {\bibinfo {volume} {289}},\ \bibinfo {pages} {597} (\bibinfo {year}
  {2009})}\BibitemShut {NoStop}%
\bibitem [{\citenamefont {Carollo}\ \emph {et~al.}(2019)\citenamefont
  {Carollo}, \citenamefont {Spagnolo}, \citenamefont {Dubkov},\ and\
  \citenamefont {Valenti}}]{Carollo2019}%
  \BibitemOpen
  \bibfield  {author} {\bibinfo {author} {\bibfnamefont {A.}~\bibnamefont
  {Carollo}}, \bibinfo {author} {\bibfnamefont {B.}~\bibnamefont {Spagnolo}},
  \bibinfo {author} {\bibfnamefont {A.~A.}\ \bibnamefont {Dubkov}}, \ and\
  \bibinfo {author} {\bibfnamefont {D.}~\bibnamefont {Valenti}},\ }\href
  {\doibase 10.1088/1742-5468/ab3ccb} {\bibfield  {journal} {\bibinfo
  {journal} {J. Stat. Mech. Theory Exp.}\ }\textbf {\bibinfo {volume} {2019}},\
  \bibinfo {pages} {094010} (\bibinfo {year} {2019})}\BibitemShut {NoStop}%
\bibitem [{\citenamefont {Hayashi}(2017)}]{Hayashi2017c}%
  \BibitemOpen
  \bibfield  {author} {\bibinfo {author} {\bibfnamefont {M.}~\bibnamefont
  {Hayashi}},\ }\href {\doibase 10.1007/978-3-662-49725-8} {\emph {\bibinfo
  {title} {{Quantum Information Theory}}}}\ (\bibinfo  {publisher} {Springer},\
  \bibinfo {address} {Berlin, Heidelberg},\ \bibinfo {year} {2017})\BibitemShut
  {NoStop}%
\bibitem [{\citenamefont {Seveso}\ \emph {et~al.}(2019)\citenamefont {Seveso},
  \citenamefont {Albarelli}, \citenamefont {Genoni},\ and\ \citenamefont
  {Paris}}]{Seveso2019}%
  \BibitemOpen
  \bibfield  {author} {\bibinfo {author} {\bibfnamefont {L.}~\bibnamefont
  {Seveso}}, \bibinfo {author} {\bibfnamefont {F.}~\bibnamefont {Albarelli}},
  \bibinfo {author} {\bibfnamefont {M.~G.}\ \bibnamefont {Genoni}}, \ and\
  \bibinfo {author} {\bibfnamefont {M.~G.~A.}\ \bibnamefont {Paris}},\ }\href
  {\doibase 10.1088/1751-8121/ab599b} {\bibfield  {journal} {\bibinfo
  {journal} {J. Phys. A}\ }\textbf {\bibinfo {volume} {in press}},\ \bibinfo
  {pages} {{}} (\bibinfo {year} {2019})}\BibitemShut {NoStop}%
\bibitem [{\citenamefont {Heinosaari}\ and\ \citenamefont
  {Ziman}(2011)}]{Heinosaari2011a}%
  \BibitemOpen
  \bibfield  {author} {\bibinfo {author} {\bibfnamefont {T.}~\bibnamefont
  {Heinosaari}}\ and\ \bibinfo {author} {\bibfnamefont {M.}~\bibnamefont
  {Ziman}},\ }\href {\doibase 10.1017/CBO9781139031103} {\emph {\bibinfo
  {title} {{The Mathematical language of Quantum Theory}}}}\ (\bibinfo
  {publisher} {Cambridge University Press},\ \bibinfo {address} {Cambridge},\
  \bibinfo {year} {2011})\BibitemShut {NoStop}%
\bibitem [{\citenamefont {Lehmann}\ and\ \citenamefont
  {Casella}(1998)}]{lehmann_theory_1998}%
  \BibitemOpen
  \bibfield  {author} {\bibinfo {author} {\bibfnamefont {E.~L.}\ \bibnamefont
  {Lehmann}}\ and\ \bibinfo {author} {\bibfnamefont {G.}~\bibnamefont
  {Casella}},\ }\href@noop {} {\emph {\bibinfo {title} {{Theory of point
  estimation}}}},\ \bibinfo {edition} {2nd}\ ed.,\ Springer texts in
  statistics\ (\bibinfo  {publisher} {Springer},\ \bibinfo {address} {New
  York},\ \bibinfo {year} {1998})\BibitemShut {NoStop}%
\bibitem [{\citenamefont {Tsiatis}(2006)}]{tsiatis06}%
  \BibitemOpen
  \bibfield  {author} {\bibinfo {author} {\bibfnamefont {A.}~\bibnamefont
  {Tsiatis}},\ }\href {\doibase 10.1007/0-387-37345-4} {\emph {\bibinfo {title}
  {{Semiparametric Theory and Missing Data}}}},\ Springer Series in Statistics\
  (\bibinfo  {publisher} {Springer},\ \bibinfo {address} {New York},\ \bibinfo
  {year} {2006})\BibitemShut {NoStop}%
\bibitem [{\citenamefont {Stoica}\ and\ \citenamefont
  {Marzetta}(2001)}]{Stoica2001}%
  \BibitemOpen
  \bibfield  {author} {\bibinfo {author} {\bibfnamefont {P.}~\bibnamefont
  {Stoica}}\ and\ \bibinfo {author} {\bibfnamefont {T.~L.}\ \bibnamefont
  {Marzetta}},\ }\href {\doibase 10.1109/78.890346} {\bibfield  {journal}
  {\bibinfo  {journal} {IEEE Trans. Signal Process.}\ }\textbf {\bibinfo
  {volume} {49}},\ \bibinfo {pages} {87} (\bibinfo {year} {2001})}\BibitemShut
  {NoStop}%
\bibitem [{\citenamefont {Hayashi}(2005)}]{Hayashi2005}%
  \BibitemOpen
  \bibfield  {author} {\bibinfo {author} {\bibfnamefont {M.}~\bibnamefont
  {Hayashi}},\ }\href {\doibase 10.1142/5630} {\emph {\bibinfo {title}
  {{Asymptotic Theory of Quantum Statistical Inference}}}},\ edited by\
  \bibinfo {editor} {\bibfnamefont {M.}~\bibnamefont {Hayashi}}\ (\bibinfo
  {publisher} {World Scientific},\ \bibinfo {address} {Singapore},\ \bibinfo
  {year} {2005})\BibitemShut {NoStop}%
\bibitem [{\citenamefont {Ragy}\ \emph {et~al.}(2016)\citenamefont {Ragy},
  \citenamefont {Jarzyna},\ and\ \citenamefont
  {Demkowicz-Dobrza{\'{n}}ski}}]{Ragy2016}%
  \BibitemOpen
  \bibfield  {author} {\bibinfo {author} {\bibfnamefont {S.}~\bibnamefont
  {Ragy}}, \bibinfo {author} {\bibfnamefont {M.}~\bibnamefont {Jarzyna}}, \
  and\ \bibinfo {author} {\bibfnamefont {R.}~\bibnamefont
  {Demkowicz-Dobrza{\'{n}}ski}},\ }\href {\doibase 10.1103/PhysRevA.94.052108}
  {\bibfield  {journal} {\bibinfo  {journal} {Phys. Rev. A}\ }\textbf {\bibinfo
  {volume} {94}},\ \bibinfo {pages} {052108} (\bibinfo {year}
  {2016})}\BibitemShut {NoStop}%
\bibitem [{\citenamefont {Belavkin}\ and\ \citenamefont
  {Grishanin}(1973)}]{Belavkin1973}%
  \BibitemOpen
  \bibfield  {author} {\bibinfo {author} {\bibfnamefont {V.~P.}\ \bibnamefont
  {Belavkin}}\ and\ \bibinfo {author} {\bibfnamefont {B.~A.}\ \bibnamefont
  {Grishanin}},\ }\href
  {https://www.maths.nottingham.ac.uk/plp/vpb/publications/Belavkin&Grishanin.pdf}
  {\bibfield  {journal} {\bibinfo  {journal} {Probl. Peredachi Inf.}\ }\textbf
  {\bibinfo {volume} {9}},\ \bibinfo {pages} {44} (\bibinfo {year}
  {1973})}\BibitemShut {NoStop}%
\bibitem [{\citenamefont {Gill}\ and\ \citenamefont {Massar}(2000)}]{Gill2000}%
  \BibitemOpen
  \bibfield  {author} {\bibinfo {author} {\bibfnamefont {R.~D.}\ \bibnamefont
  {Gill}}\ and\ \bibinfo {author} {\bibfnamefont {S.}~\bibnamefont {Massar}},\
  }\href {\doibase 10.1103/PhysRevA.61.042312} {\bibfield  {journal} {\bibinfo
  {journal} {Phys. Rev. A}\ }\textbf {\bibinfo {volume} {61}},\ \bibinfo
  {pages} {042312} (\bibinfo {year} {2000})}\BibitemShut {NoStop}%
\bibitem [{\citenamefont {Serafini}(2017)}]{serafini2017quantum}%
  \BibitemOpen
  \bibfield  {author} {\bibinfo {author} {\bibfnamefont {A.}~\bibnamefont
  {Serafini}},\ }\href
  {https://www.crcpress.com/Quantum-Continuous-Variables-A-Primer-of-Theoretical-Methods/Serafini/p/book/9781482246346}
  {\emph {\bibinfo {title} {{Quantum continuous variables : a primer of
  theoretical methods}}}}\ (\bibinfo  {publisher} {CRC Press},\ \bibinfo
  {address} {Boca Raton},\ \bibinfo {year} {2017})\BibitemShut {NoStop}%
\bibitem [{\citenamefont {Genoni}\ \emph {et~al.}(2016)\citenamefont {Genoni},
  \citenamefont {Lami},\ and\ \citenamefont {Serafini}}]{Genoni2016}%
  \BibitemOpen
  \bibfield  {author} {\bibinfo {author} {\bibfnamefont {M.~G.}\ \bibnamefont
  {Genoni}}, \bibinfo {author} {\bibfnamefont {L.}~\bibnamefont {Lami}}, \ and\
  \bibinfo {author} {\bibfnamefont {A.}~\bibnamefont {Serafini}},\ }\href
  {\doibase 10.1080/00107514.2015.1125624} {\bibfield  {journal} {\bibinfo
  {journal} {Contemp. Phys.}\ }\textbf {\bibinfo {volume} {57}},\ \bibinfo
  {pages} {331} (\bibinfo {year} {2016})}\BibitemShut {NoStop}%
\bibitem [{\citenamefont {Monras}(2013)}]{Monras2013}%
  \BibitemOpen
  \bibfield  {author} {\bibinfo {author} {\bibfnamefont {A.}~\bibnamefont
  {Monras}},\ }\href {http://arxiv.org/abs/1303.3682} {\bibfield  {journal}
  {\bibinfo  {journal} {arXiv:1303.3682}\ } (\bibinfo {year}
  {2013})}\BibitemShut {NoStop}%
\bibitem [{\citenamefont {Arthurs}\ and\ \citenamefont
  {Kelly}(1965)}]{Arthurs1965}%
  \BibitemOpen
  \bibfield  {author} {\bibinfo {author} {\bibfnamefont {E.}~\bibnamefont
  {Arthurs}}\ and\ \bibinfo {author} {\bibfnamefont {J.~L.}\ \bibnamefont
  {Kelly}},\ }\href {\doibase 10.1002/j.1538-7305.1965.tb01684.x} {\bibfield
  {journal} {\bibinfo  {journal} {Bell Syst. Tech. J.}\ }\textbf {\bibinfo
  {volume} {44}},\ \bibinfo {pages} {725} (\bibinfo {year} {1965})}\BibitemShut
  {NoStop}%
\bibitem [{\citenamefont {Stenholm}(1992)}]{Stenholm1992}%
  \BibitemOpen
  \bibfield  {author} {\bibinfo {author} {\bibfnamefont {S.}~\bibnamefont
  {Stenholm}},\ }\href {\doibase 10.1016/0003-4916(92)90086-2} {\bibfield
  {journal} {\bibinfo  {journal} {Ann. Phys. (N. Y).}\ }\textbf {\bibinfo
  {volume} {218}},\ \bibinfo {pages} {233} (\bibinfo {year}
  {1992})}\BibitemShut {NoStop}%
\bibitem [{\citenamefont {Raymer}(1994)}]{Raymer1994}%
  \BibitemOpen
  \bibfield  {author} {\bibinfo {author} {\bibfnamefont {M.~G.}\ \bibnamefont
  {Raymer}},\ }\href {\doibase 10.1119/1.17657} {\bibfield  {journal} {\bibinfo
   {journal} {Am. J. Phys.}\ }\textbf {\bibinfo {volume} {62}},\ \bibinfo
  {pages} {986} (\bibinfo {year} {1994})}\BibitemShut {NoStop}%
\bibitem [{\citenamefont {Hayashi}(2003)}]{Hayashi2006a}%
  \BibitemOpen
  \bibfield  {author} {\bibinfo {author} {\bibfnamefont {M.}~\bibnamefont
  {Hayashi}},\ }\href {http://arxiv.org/abs/quant-ph/0608198} {\bibfield
  {journal} {\bibinfo  {journal} {Bull. Math. Soc. Japan}\ }\textbf {\bibinfo
  {volume} {55}},\ \bibinfo {pages} {368} (\bibinfo {year} {2003})}\BibitemShut
  {NoStop}%
\bibitem [{\citenamefont {Guţă}\ and\ \citenamefont
  {Jen{\v{c}}ov{\'{a}}}(2007)}]{Guta2007}%
  \BibitemOpen
  \bibfield  {author} {\bibinfo {author} {\bibfnamefont {M.}~\bibnamefont
  {Guţă}}\ and\ \bibinfo {author} {\bibfnamefont {A.}~\bibnamefont
  {Jen{\v{c}}ov{\'{a}}}},\ }\href {\doibase 10.1007/s00220-007-0340-1}
  {\bibfield  {journal} {\bibinfo  {journal} {Commun. Math. Phys.}\ }\textbf
  {\bibinfo {volume} {276}},\ \bibinfo {pages} {341} (\bibinfo {year}
  {2007})}\BibitemShut {NoStop}%
\bibitem [{\citenamefont {Fujiwara}\ and\ \citenamefont
  {Yamagata}(2018)}]{Fujiwara2018}%
  \BibitemOpen
  \bibfield  {author} {\bibinfo {author} {\bibfnamefont {A.}~\bibnamefont
  {Fujiwara}}\ and\ \bibinfo {author} {\bibfnamefont {K.}~\bibnamefont
  {Yamagata}},\ }\href {http://arxiv.org/abs/1804.03510} {\bibfield  {journal}
  {\bibinfo  {journal} {arXiv:1804.03510}\ } (\bibinfo {year}
  {2018})}\BibitemShut {NoStop}%
\bibitem [{\citenamefont {Banerjee}\ \emph {et~al.}(1973)\citenamefont
  {Banerjee}, \citenamefont {Rao},\ and\ \citenamefont {Mitra}}]{Banerjee1973}%
  \BibitemOpen
  \bibfield  {author} {\bibinfo {author} {\bibfnamefont {K.~S.}\ \bibnamefont
  {Banerjee}}, \bibinfo {author} {\bibfnamefont {C.~R.}\ \bibnamefont {Rao}}, \
  and\ \bibinfo {author} {\bibfnamefont {S.~K.}\ \bibnamefont {Mitra}},\ }\href
  {\doibase 10.2307/1266840} {\bibfield  {journal} {\bibinfo  {journal}
  {Technometrics}\ }\textbf {\bibinfo {volume} {15}},\ \bibinfo {pages} {197}
  (\bibinfo {year} {1973})}\BibitemShut {NoStop}%
\bibitem [{\citenamefont {Bickel}\ \emph {et~al.}(1993)\citenamefont {Bickel},
  \citenamefont {Klaassen}, \citenamefont {Ritov},\ and\ \citenamefont
  {Wellner}}]{bickel93}%
  \BibitemOpen
  \bibfield  {author} {\bibinfo {author} {\bibfnamefont {P.~J.}\ \bibnamefont
  {Bickel}}, \bibinfo {author} {\bibfnamefont {C.~A.~J.}\ \bibnamefont
  {Klaassen}}, \bibinfo {author} {\bibfnamefont {Y.}~\bibnamefont {Ritov}}, \
  and\ \bibinfo {author} {\bibfnamefont {J.~A.}\ \bibnamefont {Wellner}},\
  }\href@noop {} {\emph {\bibinfo {title} {{Efficient and Adaptive Estimation
  for Semiparametric Models}}}}\ (\bibinfo  {publisher} {Springer},\ \bibinfo
  {address} {New York},\ \bibinfo {year} {1993})\BibitemShut {NoStop}%
\bibitem [{\citenamefont {Luenberger}(1997)}]{luenberger1997optimization}%
  \BibitemOpen
  \bibfield  {author} {\bibinfo {author} {\bibfnamefont {D.~G.}\ \bibnamefont
  {Luenberger}},\ }\href {https://books.google.pl/books?id=lZU0CAH4RccC} {\emph
  {\bibinfo {title} {{Optimization by Vector Space Methods}}}}\ (\bibinfo
  {publisher} {Wiley},\ \bibinfo {year} {1997})\BibitemShut {NoStop}%
\bibitem [{\citenamefont {Genoni}\ \emph {et~al.}(2013)\citenamefont {Genoni},
  \citenamefont {Paris}, \citenamefont {Adesso}, \citenamefont {Nha},
  \citenamefont {Knight},\ and\ \citenamefont {Kim}}]{Genoni2013b}%
  \BibitemOpen
  \bibfield  {author} {\bibinfo {author} {\bibfnamefont {M.~G.}\ \bibnamefont
  {Genoni}}, \bibinfo {author} {\bibfnamefont {M.~G.~A.}\ \bibnamefont
  {Paris}}, \bibinfo {author} {\bibfnamefont {G.}~\bibnamefont {Adesso}},
  \bibinfo {author} {\bibfnamefont {H.}~\bibnamefont {Nha}}, \bibinfo {author}
  {\bibfnamefont {P.~L.}\ \bibnamefont {Knight}}, \ and\ \bibinfo {author}
  {\bibfnamefont {M.~S.}\ \bibnamefont {Kim}},\ }\href {\doibase
  10.1103/PhysRevA.87.012107} {\bibfield  {journal} {\bibinfo  {journal} {Phys.
  Rev. A}\ }\textbf {\bibinfo {volume} {87}},\ \bibinfo {pages} {012107}
  (\bibinfo {year} {2013})}\BibitemShut {NoStop}%
\bibitem [{\citenamefont {Ragy}\ \emph {et~al.}(2019)\citenamefont {Ragy},
  \citenamefont {Jarzyna},\ and\ \citenamefont
  {Demkowicz-Dobrza{\'{n}}ski}}]{Ragy2019}%
  \BibitemOpen
  \bibfield  {author} {\bibinfo {author} {\bibfnamefont {S.}~\bibnamefont
  {Ragy}}, \bibinfo {author} {\bibfnamefont {M.}~\bibnamefont {Jarzyna}}, \
  and\ \bibinfo {author} {\bibfnamefont {R.}~\bibnamefont
  {Demkowicz-Dobrza{\'{n}}ski}},\ }\href {\doibase 10.1103/PhysRevA.99.029905}
  {\bibfield  {journal} {\bibinfo  {journal} {Phys. Rev. A}\ }\textbf {\bibinfo
  {volume} {99}},\ \bibinfo {pages} {029905} (\bibinfo {year}
  {2019})}\BibitemShut {NoStop}%
\bibitem [{\citenamefont {Bhatia}(1997)}]{Bhatia1997}%
  \BibitemOpen
  \bibfield  {author} {\bibinfo {author} {\bibfnamefont {R.}~\bibnamefont
  {Bhatia}},\ }\href@noop {} {\emph {\bibinfo {title} {{Matrix Analysis}}}}\
  (\bibinfo  {publisher} {Springer},\ \bibinfo {address} {New York},\ \bibinfo
  {year} {1997})\BibitemShut {NoStop}%
\end{thebibliography}%
\end{document}